\documentclass[sigconf]{aamas} 

\usepackage{balance} 
\usepackage{caption}
\usepackage{subcaption}

\usepackage{algorithmic}
\usepackage[linesnumbered,ruled,vlined]{algorithm2e}

\usepackage{dblfloatfix}

\usepackage{amsmath,amsfonts,bm,bbm}

\newtheorem{innercustomgeneric}{\customgenericname}
\providecommand{\customgenericname}{}
\newcommand{\newcustomtheorem}[2]{%
  \newenvironment{#1}[1]
  {%
   \renewcommand\customgenericname{#2}%
   \renewcommand\theinnercustomgeneric{##1}%
   \innercustomgeneric
  }
  {\endinnercustomgeneric}
}

\newcustomtheorem{customthm}{Theorem}
\newcustomtheorem{customlemma}{Lemma}

\def\eqref#1{equation~(\ref{#1})}

\def\1{\bm{1}}

\usepackage{cancel}

\usepackage{listings}
\usepackage{enumerate}

\DeclareMathAlphabet{\mathsfit}{\encodingdefault}{\sfdefault}{m}{sl}
\SetMathAlphabet{\mathsfit}{bold}{\encodingdefault}{\sfdefault}{bx}{n}

\newcommand{\E}{\mathbb{E}}

\newcommand{\R}{\mathbb{R}}

\newcommand{\argmax}{\operatornamewithlimits{argmax}}

\usepackage{amsthm}

\makeatletter
\newtheorem*{rep@theorem}{\rep@title}
\newcommand{\newreptheorem}[2]{%
\newenvironment{rep#1}[1]{%
 \def\rep@title{#2 \ref{##1}}%
 \begin{rep@theorem}}%
 {\end{rep@theorem}}}
\makeatother

\newtheorem{theorem}{Theorem}
\newreptheorem{theorem}{Theorem}

\newreptheorem{proposition}{Proposition}

\newreptheorem{corollary}{Corollary}

\newreptheorem{lemma}{Lemma}

\newtheorem{definition}{Definition}
\newtheorem{example}{Example}

\newcommand{\bE}{\mathbb{E}}

\newcommand{\bb}{\mathbf{b}}
\newcommand{\bd}{\mathbf{d}}

\newcommand{\bA}{\mathbf{A}}
\newcommand{\bC}{\mathbf{C}}
\newcommand{\bN}{\mathbf{N}}
\newcommand{\bM}{\mathbf{M}}

\newcommand{\bx}{\mathbf{x}}
\newcommand{\by}{\mathbf{y}}
\newcommand{\bw}{\mathbf{w}}
\newcommand{\bv}{\mathbf{v}}
\newcommand{\bX}{\mathbf{X}}

\newcommand{\btheta}{\boldsymbol\theta}

\newcommand{\cD}{\mathcal{D}}

\newcommand{\cC}{\mathcal{C}}

\newcommand{\cP}{\mathcal{P}}

\newcommand{\cY}{\mathcal{Y}}

\newcommand{\defword}[1]{\textbf{\boldmath{#1}}}
\newcommand{\ie}{{\it i.e.},~}  
\newcommand{\eg}{{\it e.g.},~}  

\newtheorem*{theorem-nonumber}{Theorem}
\newtheorem*{lemma-nonumber}{Lemma}

\definecolor{darkgreen}{RGB}{0,125,0}
\definecolor{darkblue}{RGB}{0,0,125}

\newcounter{mlNoteCounter}

\newcounter{klNoteCounter}

\newcounter{gpilNoteCounter}

\newcounter{imgempNoteCounter}

\newcounter{mdNoteCounter}

\setcopyright{ifaamas}
\acmConference[AAMAS '25]{Proc.\@ of the 24th International Conference
on Autonomous Agents and Multiagent Systems (AAMAS 2025)}{May 19 -- 23, 2025}
{Detroit, Michigan, USA}{A.~El~Fallah~Seghrouchni, Y.~Vorobeychik, S.~Das, A.~Nowe (eds.)}
\copyrightyear{2025}
\acmYear{2025}
\acmDOI{}
\acmPrice{}
\acmISBN{}

\title[Soft Condorcet Optimization]{Soft Condorcet Optimization for Ranking of General Agents}

\author{Marc Lanctot}
\affiliation{
  \institution{Google DeepMind}
  \city{Montreal}
  \country{Canada}
  }

\author{Kate Larson}
\affiliation{
  \institution{Google DeepMind,\\ University of Waterloo}
  \city{Waterloo}
  \country{Canada}
  }

\author{Michael Kaisers}
\affiliation{
  \institution{Google DeepMind}
  \city{Paris}
  \country{France}
}

\author{Quentin Berthet}
\affiliation{
  \institution{Google DeepMind}
  \city{Paris}
  \country{France}
}

\author{Ian Gemp}
\affiliation{
  \institution{Google DeepMind}
  \city{London}
  \country{United Kingdom}
}

\author{Manfred Diaz}
\affiliation{
  \institution{Mila, University of Montreal}
  \city{Montreal}
  \country{Canada}
}

\author{Roberto-Rafael Maura-Rivero}
\affiliation{
  \institution{Google DeepMind, LSE}
  \city{London}
  \country{United Kingdom}
}

\author{Yoram Bachrach}
\affiliation{
  \institution{Meta$^1$}
  \city{London}
  \country{United Kingdom}
}

\author{Anna Koop}
\affiliation{
  \institution{Google DeepMind}
  \city{Montreal}
  \country{Canada}
}

\author{Doina Precup}
\affiliation{
  \institution{Google DeepMind}
  \city{Montreal}
  \country{Canada}
}

\begin{abstract}
Driving progress of AI models and agents requires comparing their performance on standardized benchmarks;
for {\it general} agents, individual performances must be aggregated across a potentially wide variety of different tasks.
In this paper, we describe a novel ranking scheme inspired by social choice frameworks, called Soft Condorcet Optimization (SCO), to compute the optimal ranking of agents: the one that makes the fewest mistakes in predicting the agent comparisons in the evaluation data.
This optimal ranking is the maximum likelihood estimate  when evaluation data (which we view as votes) are interpreted as noisy samples from a ground truth ranking, a solution to Condorcet's original voting system criteria.
SCO ratings are maximal for Condorcet winners when they exist, which we show is not necessarily true for the classical rating system Elo. 
We propose three optimization algorithms to compute SCO ratings and evaluate their empirical performance.
When serving as an approximation to the Kemeny-Young voting method, SCO rankings are on average 0 to 0.043 away from the optimal ranking in normalized Kendall-tau distance across 865 preference profiles from the PrefLib open ranking archive. 
In a simulated noisy tournament setting, SCO achieves accurate approximations to the ground truth ranking and the best among several baselines when 59\% or more of the preference data is missing. Finally, SCO ranking provides the best approximation to the optimal ranking, measured on held-out test sets, in a problem containing 52,958 human players across 31,049 games of the classic seven-player game of Diplomacy.
\end{abstract}


\newcommand{\BibTeX}{\rm B\kern-.05em{\sc i\kern-.025em b}\kern-.08em\TeX}

\begin{document}


\pagestyle{fancy}
\fancyhead{}


\maketitle 


\section{Introduction}

\footnotetext[1]{Work performed while at Google DeepMind.}

Progress in the field of artificial intelligence has been driven by measuring the performance of agents on common benchmarks and challenge problems~\cite{Samuel59,DeepBlue,tdgammon,arimaa,Silver16Go,Vinyals19Alphastar,Cicero}. 
In machine learning, common benchmarks like the UCI data set repository allowed direct comparisons of supervised learning algorithms~\cite{UCI}. Competitions, such as ImageNet, led to breakthroughs in deep learning~\cite{imagenet}. 

All of these examples require comparing agents (or models). Original success stories such as DeepBlue, TD-Gammon, and AlphaGo focused on a single domain. In the past ten years, agents have become increasingly more generally capable.
AlphaZero extended application of AlphaGo to chess and Shogi~\cite{Silver18AlphaZero}. 
The Arcade Learning Environment~\cite{bellemare13arcade}, which steered much of the agent development in deep reinforcement learning, evaluated agents across 57 different Atari games. 
Recently, language models have been evaluated across suites of tasks such as in HELM~\cite{liang2022holistic}, BIG-bench~\cite{srivastava2023beyond} AgentBench~\cite{liu2023agentbench}, and via a public leaderboard such as Chatbot Arena driven by human voting~\cite{chiang24chatbot}.
Answering simple questions for these generally capable agents, such as ``Which is the best agent?'' or ``Is agent $X$ better than agent $Y$?'' or ``What is the relative ranking of agents $X$, $Y$, and $Z$?'' become increasingly more difficult when aggregating over many different contexts:
how agents are scored may vary wildly across tasks, 
data collected for evaluation may not be balanced evenly across tasks (or agents), and classical rating systems were simply not designed for this use case. 

To address these problems, recent ranking methods such as Vote'N'Rank~\cite{rofin-etal-2023-votenrank} and Voting-as-Evaluation (VasE)~\cite{lanctot2023evaluating} use voting methods to aggregate results across tasks. 
Using computational social choice as a basis for ranking agents has several benefits: well-studied consistency properties of the voting methods are inherited, they do no require score normalization, and are less sensitive to score values and agent population than game-theoretic evaluation schemes~\cite{Balduzzi18ReEval}.
However, classic voting schemes, and related tournament solutions,  typically assume that the data (e.g. agent comparisons) is complete. This assumption is not necessarily valid in the agent evaluation setting. While there has been research on identifying “necessary and possible winners”  when there is incomplete voting or comparison data, the results are mixed~\cite{pini2011possible,xia2011determining,aziz2015possible} and most of the findings are focused  on identifying top-ranked agents, not ranking  all agents as is our focus.  

In this paper, we introduce a new ranking scheme for general agents inspired by the interpretation of voting rules as maximum likelihood estimators~\cite{ConitzerSandholm05}.
Starting with Condorcet's original model of voting~\cite[Chapter 8]{Brandt16Handbook}, Young showed that the maximum likelihood estimate (MLE) of the true ranking is the one that minimizes the sum of Kendall-tau distances to all the votes.
Soft Condorcet Optimization (SCO) solves an optimization problem
that assigns a numerical {\it rating} (score) $\theta_a$ to each agent (alternative) $a$. 
SCO then treats these ratings as a parameter vector, the votes as a data set, and defines a differentiable loss function as the objective, which can be optimized in several different ways.
The final ranking of agents is obtained by sorting the ratings.

In summary, this paper makes the following contributions.

\noindent {\bf 1. SCO ranking scheme} with the following properties:
(1a) Three optimization methods to find ratings and corresponding rankings: gradient descent applied to a soft Kendall-tau distance (``sigmoid loss''), or a Fenchel-Young loss (perturbed optimization)~\cite{Berthet20}; or solving a sigmoidal program with a branch-and-bound method~\cite{Udell2014MaximizingAS}.
(1b) Online form that can update ratings, and thus rankings, from individual outcomes as evaluation data arrives.
(1c) Theorem~\ref{th:sigmoid-condorcet} guarantees that the top-ranked agent by SCO ratings according to the sigmoid loss is the Condorcet winner when one exists.

\noindent {\bf 2. Empirical evaluations} that demonstrate the following:
(2a) SCO ranking using sigmoid loss solves a failure mode of classical Elo rating system which may top-rank an agent that is not a Condorcet winner even when one exists.
(2b) SCO can serve as an approximation to the Kemeny-Young voting method, indeed empirically finding low approximation error to the optimal ranking: on average 0 to 0.043 away in normalized Kendall-tau distance across 865 preference profiles from the PrefLib~\cite{Preflib_MaWa13a}.
(2c) In a noisy tournament setting with sparse data, SCO approximates the true ranking best when a large proportion (59\% or more) of the data is missing.
(2d) SCO ratings are closer to optimal rankings than Elo and voting-as-evaluation methods on held out test sets over 31,049 human Diplomacy games played by 52,958 players.

\section{Background}

In this section, we describe the building blocks and introduce some basic terminology required to understand our method.

\subsection{Evaluation of General Agents}

We first paraphrase several key descriptions from~\cite{Balduzzi18ReEval,lanctot2023evaluating}.
The problem of evaluating agents is that of ranking agents according to their skill. 
Skill can be determined in several ways.
In the {\bf Agent-versus-Task (AvT)} setting, agents compete individually in different tasks and compare outcomes (scores) to each other in each task (\eg language models and the various metrics for assessing their abilities).
In the {\bf Agent-versus-Agent (AvA)} setting, agents directly compete against each other (\eg online games such as chess or Diplomacy).

\subsubsection{Classical Evaluation}
\label{sec:background-classical-eval}

Elo is a classic rating system that uses a simple logistic model learned from win/loss/draw outcomes~\citep{Elo78}. 
A rating, $r_i$, is assigned to each player $i$ such that the probability of player $i$ beating player $j$ is predicted as $\hat{p}_{ij} = \frac{1}{1 + 10^{(r_j - r_i)/400}}$. 
While Elo was designed specifically to rate players in the two-player zero-sum, perfect information game of Chess, it has been widely applied to other domains, including in evaluation of large language models~\citep{zheng2023judging}.
TrueSkill~\citep{TrueSkill06} and bayeselo~\citep{Coulom05bayeselo} are rating systems based on similar foundations (Bradley-Terry models of skill) that also model uncertainty over ratings using Bayesian methods.

Elo has a number of positive qualities. First, it is a simple rule. Second, it can be used to estimate win rates between any two agents. Third, it can be easily employed {\it online}, \ie to modify players' ratings from the result of a single game. It is also a special case of logistic regression (for details, see 
Appendix~\ref{sec:relationship-to-elo}). Hence, Elo has been widely applied and is a common default choice for evaluation of agents.
However, Elo has a number of well-known drawbacks~\cite{Balduzzi18ReEval,lanctot2023evaluating, bertrand2023elolim}. Of particular interest is the incompatibility of Elo with the concept of a Condorcet winner from social choice theory; examples are summarized in Section~\ref{sec:eval-warmup}.

\subsubsection{Voting as Evaluation of General Agents}
\label{sec:background-vase}

Another way to evaluate agents is to use social choice theory, called Voting-as-Evaluation (VasE) recently proposed in~\cite{lanctot2023evaluating}. In VasE, the alternatives correspond to general agents and votes to assessments of their performance. In the AvT setting, agents are ordered based on their performance on different tasks (such as each game in the Atari Learning Environment~\cite{bellemare13arcade} or on different benchmarks for large language models~\cite{liu2023agentbench}). In the AvA setting, agents compete directly in a multiagent environment, such as multiplayer games (like chess or poker), and each game outcome corresponds to a ranking over a subset of agents. Chatbot Arena~\cite{zheng2023judging}, where language models compete head-to-head to provide the best answer to the same question, is another instance of the AvA setting. Casting agent evaluation as an application of computational social choice provides the benefit of robustness in the form of Condorcet consistency, clone/composition consistency, agenda consistency, and/or population consistency depending on the choice of voting rule used to evaluate agents. 
However, it is unclear how well these methods perform when the data is missing or unevenly distributed, which often occurs in the agent evaluation setting.
SCO is motivated similarly to VasE and, as such, will be empirically assessed mainly for agent evaluation. To make the connection to ideas from the social choice literature, we will often use voting language when discussing SCO and its use in evaluation. Relevant terminology is introduced later in the paper. 

\subsection{Permutation and Ranking Distances}

We now define distance metrics over rankings that will play a key role when describing our method and loss function.
Informally, the Kendall-tau distance counts the number of pairwise disagreements between permutations. 

\begin{definition}
Let $S_1, S_2$ be finite sets of elements such that $S_1 \subseteq S_2$.
Let $\pi_1, \pi_2$ be permutations over elements in $S_1$ and $S_2$, respectively.
The \defword{Kendall-tau distance} between two permutations is defined as 
\begin{equation}
K_d(\pi_1, \pi_2) = \sum_{\{i,j\} \in \cC_2(S_1)} \bar{K}_{i,j}(\pi_1, \pi_2),
\end{equation}
where $\cC_2(S)$ is the set of unordered pairs of $S$ (combinations of size 2), 
$\bar{K}_{i,j}(\pi_1, \pi_2) = 0$ if $i$ and $j$ are in the same order in $\pi_1$ and $\pi_2$, and $\bar{K}_{i,j}(\pi_1, \pi_2) = 1$ otherwise.
\end{definition}

Note that this definition allows one set to be a subset of the other, which is more general than the standard definition; this distinction is necessary for the evaluation metric used in Section~\ref{sec:eval-diplomacy}, and corresponds to the standard definition when $S_1 = S_2$.

Since the maximum distance is ${|S_1| \choose 2} = \frac{|S_1| (|S_1| - 1)}{2}$, this can be easily normalized to be in $[0,1]$:

\begin{definition}
Let $S_1, S_2$ be finite sets of elements such that $S_1 \subseteq S_2$.
Let $\pi_1, \pi_2$ be permutations over elements in $S_1$ and $S_2$, respectively.
The \defword{normalized Kendall-tau distance} $\pi_1$ and $\pi_2$ is defined as
\begin{equation}
K_n(\pi_1, \pi_2) = \frac{2 K_d(\pi_1, \pi_2)}{|S_1| (|S_1| - 1)}.
\end{equation}
\end{definition}

\subsection{Social Choice Theory}
\label{sec:background-sct}

\begin{figure}[t!]
    \centering
    \includegraphics[width=0.38\textwidth]{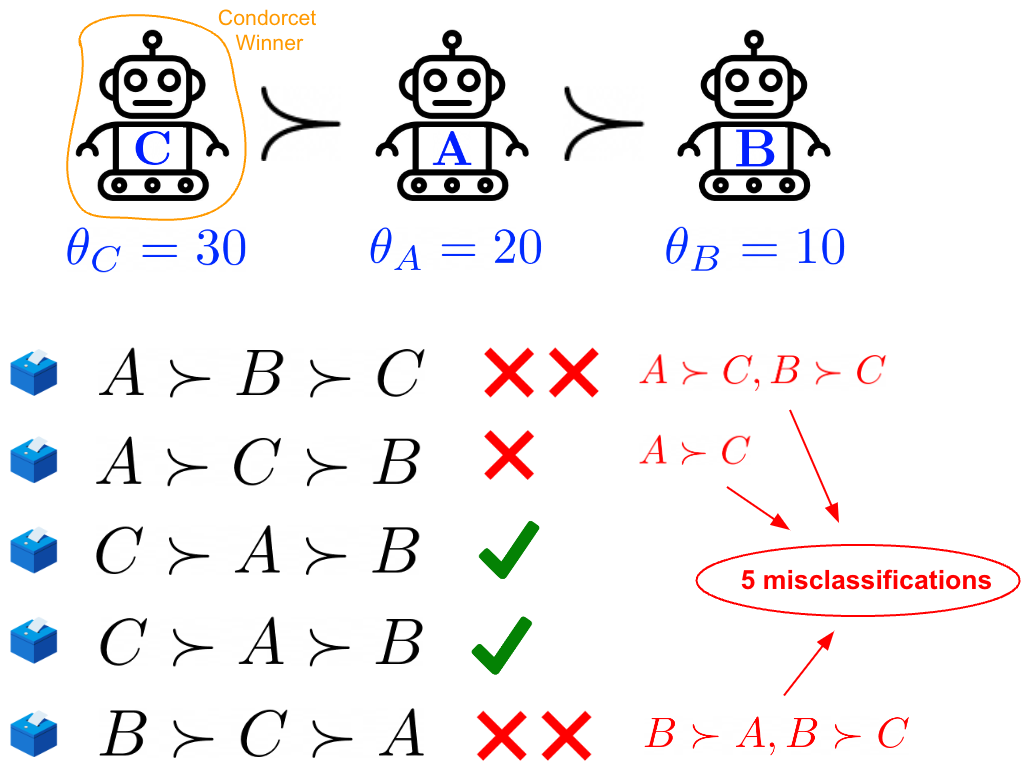}
    \caption{An example calculation of the Kendall-tau distances of each vote in the preference profile from Table~\ref{tab:example} to the optimal ranking $R = C \succ A \succ B$. The sum of the distances is 5. This sum for any other ranking $R' \not= R$ is greater than 5.}
    \label{fig:example-fig}
\end{figure}

A \defword{voting scheme} is defined as $\langle A, V, f\rangle$ 
where $A=\{a_1,\ldots, a_m\}$ is the set of \defword{alternatives} (agents),  
$V=\{v_1,\ldots, v_n\}$ is the set of \defword{voters}, 
and $f$ is the \defword{voting rule} that determines how votes are aggregated.
Voters have \defword{preferences} over alternatives: $a_1\succ_{v_i} a_2$ indicates voter $v_i$ strictly prefers alternative $a_1$ over alternative $a_2$.
In this paper, we assume strict preferences only, however the ideas can easily be extended to the case that allow weak preference (including ties/indifference between alternatives, \eg $a_1 \succeq s_2$).
These preferences induce total orders over alternatives, which we denote by $\mathcal{L}$.  
A \defword{preference profile}, $[\succ]\in \mathcal{L}^n$, is a vector specifying the preferences of each voter in $V$. 
It can be useful to summarize the preference profile in a \defword{voter preference matrix} $\bN$ or \defword{vote margin matrix} $\bM$.
The preference count $N(x,y)$, $x, y \in A$ is the number of voters in $[\succ]$ that strictly prefer $x$ to $y$. 
The vote margin is the difference in preference count: $\delta(x, y) = N(x,y) - N(y,x)$.
Table~\ref{tab:example} shows a preference profile and its resulting preference matrix, $\bN = (N(i,j))_{i,j}$, and margin matrix $\bM = (\delta(i,j))_{i,j} = \bN - \bN^T$.

The central problem of social choice theory is how to aggregate preferences of a population so as to reach some collective decision.
A voting rule that determines the ``winner'' (a non-empty subset, possibly with ties), is a \defword{social choice function} (SCF). A voting rule that returns an aggregate ranking (total order) over all the alternatives is a \defword{social welfare function} (SWF).
Much of the social choice literature focuses on understanding what properties different voting rules support.

The \defword{Condorcet winner} defines a fairly intuitive concept: $A$ is the (strong) Condorcet winner if the number of votes where $A$ is ranked higher than $B$ is greater than vice versa for all other alternatives $B$.
A weak Condorcet winner wins or ties in every head-to-head pairing.
Formally, given a preference profile, $[\succ]$, 
a weak Condorcet winner is an alternative $a^* \in A$ such that $\forall a' \in A, \delta(a^*, a') \ge 0$. 
If the inequality is strict for all pairs except  $(a^*, a^*)$ then we call it a strong Condorcet winner.
In the example shown in Table~\ref{tab:example}, alternative $C$ is the strong Condorcet winner. It dominates $A$ since three out of the five voters prefer $C$ to $A$. A similar situation holds when $C$ is compared to alternative $B$. 
While many have argued that this definition captures the essence of the correct collective choice~\citep{deCondorcet1785}, 
in practice  preference profiles may have no Condorcet winner. Condorcet-consistent voting schemes (\eg Kemeny-Young introduced next) return a Condorcet winner when it exists, but differ on how they handle settings with no Condorcet winners.

\begin{table}[t]
\begin{tabular}{ll}
\hline
1: & $A \succ B \succ C$\\
1: & $A \succ C \succ B$\\
2: & $C \succ A \succ B$\\
1: & $B \succ C \succ A$\\
\hline
\end{tabular}
\hspace{0.3cm}
\begin{tabular}{|l|lll|}
\hline
  & $A$ & $B$ & $C$\\
\hline
$A$ & 0 & 4 & 2\\
$B$ & 1 & 0 & 2\\
$C$ & 3 & 3 & 0\\
\hline
\end{tabular}
\hspace{0.3cm}
\begin{tabular}{|l|lll|}
\hline
  & $A$ & $B$ & $C$\\
\hline
$A$ & 0 & 3 & -1\\
$B$ & -3 & 0 & -1\\
$C$ & 1 & 1 & 0\\
\hline
\end{tabular}
\caption{Left: an example preference profile with five votes~\cite[Figure 1]{lanctot2023evaluating}. The number on the left of the colon represents the number of votes of that type.
Middle: the voter preference matrix, $\bN$. Right: the voter margin matrix, $\bM$. 
\label{tab:example}}
\end{table}

\subsubsection{Kemeny-Young Voting Method}
\label{sec:kem-young-voting}

The voting method was initially proposed by Kemeny~\cite{Kemeny59}. Later, its properties were characterized by Young \& Levenglick ~\cite{young1978consistent}.
Let each ranking (total order) be represented as a permutation $\pi$ over $|A|$ alternatives. 
Define the Kemeny score of permutation $\pi$ as $\textsc{KemenyScore}(\pi) = \sum_{(i,j), i < j} N(\pi[i], \pi[j])$.
The Kemeny rule returns the ranking that maximizes this Kemeny score: $\argmax_{\pi \in \Pi(|A|)} \textsc{KemenyScore}(\pi)$.
Kemeny-Young is {\it Condorcet-consistent}: if a Condorcet winner exists, it will be top-ranked by Kemeny-Young. It also satisfies the majority criterion, the Smith criterion~\cite{Smith73}, and monotonicity. 
The Kemeny rule always returns an \defword{optimal ranking}, \ie one whose sum of Kendall-tau distances to the votes is minimal, but its computational complexity is prohibitively expensive when $m$ is large.

\subsection{Learning-to-Rank}
\label{sec:background-learning-to-rank}

Another related field is that of learning-to-rank~\cite{Liu11}. The canonical example is that a user enters a keyword (query) and the problem is to retrieve the most relevant documents in a database, ranked by relevance to the keyword. 
There are several algorithms that learn to rank; Google's PageRank, which powers their search engine, is one example.
Our setting can be characterized by a learning-to-rank problem with a constant keyword (or no keyword) as there is no query.
An important class of statistical models in this setting is random utility models~\cite{Xia19}.
In random utility models, assessments are perceived as some ground truth assessment plus some noise.

Perturbed optimizers~\cite{Berthet20,Blondel20} transform non-differentiable functions (such as sorting and ranking) into smoothed versions by adding noise; these perturbed functions can then be optimized by gradient descent.
This is part of a growing effort to allow end-to-end training through discrete operators, using classical stochastic smoothing and perturbation approaches \cite{gumbel1954statistical, hazan2016perturbations}.  Included are optimal transport, clustering, dynamic time-warping and other dynamic programs \citep{cuturi2013sinkhorn, cuturi2017soft, mensch2018differentiable, vlastelica2019differentiation, paulus2020gradient, sander2023fast, stewart2023differentiable} applied in fields such as computer vision, audio processing, biology, and physical simulators \citep{cordonnier2021differentiable, kumar2021groomed, carr2021self, le2021differentiable, baid2023deepconsensus, llinares2023deep} and other optimization algorithms \citep{dubois2022fast}.

\section{Soft Condorcet Optimization}
\label{sec:sco}
Soft Condorcet Optimization (SCO) is a ranking scheme for evaluation of general agents inspired by social choice theory.  
An SCO ranking is derived from and represented by {\bf SCO ratings}, $\theta_a$, for each alternative $a \in A$. 
The rating $\theta_a \in [\theta_{\min}, \theta_{\max}]$ serves only to determine $a$'s relative order compared to other alternatives in the ranking, such that $a\succ b$ if and only if $\theta_a > \theta_b$. The ranking is induced by numerically sorting these ratings (\eg see Figure~\ref{fig:example-fig}). 
We then formulate an optimization problem by carefully constructing a loss function that that penalizes discrepancies or misclassifications in the ordinal relationships between alternatives.

An SCO rating is a numerical value representing an agent's level of skill.
SCO is closely related to several prior works: Elo~\cite{Elo78}, probabilistic ranking~\cite{diaconis1988rankprob, marden1995rankdata, alvo2014rankprob}, and perturbed optimizers~\cite{Blondel20,Berthet20}.
We elaborate on these relationships in Appendix~\ref{sec:relationships-to-prior-work}.

\subsection{SCO Ratings and the Sigmoid Loss}
\label{sec:sco-grad-descent}

In this section, we first explain the {\bf sigmoid loss} function. For some set of ratings $\btheta$, this loss function quantifies the amount of error: the level of disagreement between the specific values of each rating and the data (preference profiles obtained through agent evaluation). The goal is then to find an assignment of ratings that minimizes this loss, \ie the set of ratings that best explain the preference data. 

Let $\Pi(A)$ denote the set of permutations over alternatives $A$. We will call the data we work with ``votes'' to emphasize that we are working with (partial) rankings over alternatives, and borrow terminology from preferences and social choice. In particular, we view the entire dataset to be a collection of votes and so will refer to  it as \emph{preference profile}, $[\succ]$.
Let $v \in [\succ]$ refer to each vote in the profile, where $v$ is a permutation over subsets of $A$, with length $|v|$. For a vote $v$, denote $v[i]$ as the alternative in position $i$ such that vote $v$ is represented as: $v[0] \succ v[1] \succ \cdots \succ v[|v|-1]$.

The ultimate goal is to find an {\bf optimal ranking} $R$:
\begin{definition}
\label{def:optimal-ranking}
Given some a profile $[\succ]$ (\ie set of votes), an {\bf optimal ranking} minimizes the sum of the Kendall-tau distances \[R = argmin_{R \in \Pi(A)} \sum_{v \in [\succ]} K_d(v, R).\]
\end{definition}
\noindent Let index pairs $I_2(v) = \{ (i,j)~|~i,j \in \{ 0, 1, \cdots, |v| - 1\} \mbox{ and } i < j \}$. Given a preference profile and ratings $\btheta$, we define a {\bf discrete loss}:
\begin{equation}
L([\succ], A, V, \btheta) = \sum_{v \in [\succ]} \sum_{(i,j) \in I_2(v)} D_v(\theta_{v[i]}, \theta_{v[j]}),   \label{eq:ktd-loss}
\end{equation}
where $D_v(\theta_a, \theta_b)$ is a function that measures the {\it discrepancy} between the positions of alternatives $a$ and $b$:
\begin{equation}
D_v(\theta_a, \theta_b) = \left\{ \begin{array}{ll}
                                1 & \mbox{if $\theta_b - \theta_a > 0$ in $v$}; \\
                                0 & \mbox{otherwise},\end{array} \right.,    \label{eq:discrepancy}
\end{equation}
then the minimum of the discrete loss function corresponds to a ratings assignment $\btheta$ such that ranking induced by $\btheta$  minimizes the sum of the Kendall-tau distances to all the votes.

\begin{example}
\label{eg:ktd-loss-example}
Recall the example from Figure~\ref{fig:example-fig} showed the computation of the sum of Kendall-tau distances from the ranking $R = C \succ A \succ B$ to votes depicted in preference profile in Table~\ref{tab:example}.
In this example, we show how the value is computed under rating vector:
\[ \btheta = (\theta_A, \theta_B, \theta_C) = (20, 10, 30). \]
Let $[\succ]$ be the preference profile depicted in Table~\ref{tab:example}. We now show that the main loss function (Equation~(\ref{eq:ktd-loss})) leads to same value as in Figure~\ref{fig:example-fig}.
The outer sum of Equation~\ref{eq:ktd-loss} enumerates the votes, which we will assume is in the same order as listed in Figure~\ref{fig:example-fig}.
The inner sum computes the Kendall-tau distance from the vote $v$ to the ranking $R$ induced by $\btheta$ (red exes in Figure~\ref{fig:example-fig}).
For the first vote $v = A \succ B \succ C$, the discrepancy function $D$ outputs 1 two times: once with pair $(i, j) = (0, 2)$ and once with pair $(i, j) = (1, 2)$ because the preferences between agents $(A, C)$ and agents $(B, C)$ disagree between $\btheta$ and $v$, so the inner sum for the first vote is 2.
Similarly for the other votes: they correspond precisely to the same values as in Figure~\ref{fig:example-fig}. Hence, the loss function (Equation~(\ref{eq:ktd-loss})) is simply computing the sum of Kendall-tau distances between the ranking induced by $\btheta$ and all the votes.
\end{example}

Since $D$ is a step function discontinuous at $\theta_a = \theta_b$, it is not differentiable in $\btheta$.
To solve this, we replace $D$ with a smooth approximation, \ie the logistic function
\begin{equation}
\tilde{D}_v(\theta_a, \theta_b) = \sigma(\theta_b - \theta_a) = \frac{1}{1 + e^{(\theta_a - \theta_b)/\tau}}, 
\label{eq:sigmoid-discrepancy-func}
\end{equation}
leading to a soft Kendall-tau (differentiable) \defword{sigmoid loss}:
\begin{equation}
\tilde{L}([\succ], A, V, \btheta) = \sum_{v \in [\succ]} \sum_{(i,j) \in I_2(v)} \tilde{D}_v(\theta_{v[i]}, \theta_{v[j]}). \label{eq:sco-loss}
\end{equation}
The sigmoid loss is a differentiable version of the Kendall-tau distance sum and acts as a smooth approximation to the discrete loss.

Note that while we focus on the sigmoid loss as a soft approximation to Kendall-tau distance in this paper, the same approach can be used for other ranking distances that can be approximated by differentiable functions, such as Spearman's footrule distance~\cite{Diaconis77}.

\subsection{Sigmoid Loss Minimization}
\label{sec:sco-sigmoid-loss-min}

SCO ratings can be computed using the sigmoid loss in two ways.

\subsubsection{Gradient Descent}
\label{sec:sco-via-sgd}

The most straight-forward way is to apply gradient descent~\cite{HazanOCObook}: update ratings by following the gradient of the sigmoid loss, as shown in Algorithm~\ref{alg:sco}.
After applying the gradient to the ratings on line~\ref{alg:update-line}, the ratings may escape the bounded constraint space so we  project them back. This is a straight-forward application of standard  gradient descent~\cite[Chapter 2]{HazanOCObook}.
A common variant is stochastic gradient descent (SGD) which estimates the gradient by sampling subsets, \ie ``batches'', of the data set, computing the gradient using the sampled batch only. 
The standard $\ell_2$ projection step, $\textsc{Proj}$, projects the ratings back to the hypercube by clipping any ratings that fall outside the valid range $[\theta_{min}, \theta_{max}]$.

We compute the gradient for a subset of the votes (\ie batch $B \subseteq [\succ]$), $\nabla_{\btheta} \tilde{L}(B, \btheta)$, from the batch loss $\tilde{L}(B, \btheta)$, which resembles~\eqref{eq:sco-loss} but summed only over the votes in $B$ using the continuous $\tilde{D}_v$
from~\eqref{eq:sigmoid-discrepancy-func}:
\begin{equation}
\tilde{L}(B, \btheta) = \sum_{v \in B} \sum_{(i,j) \in I_2(v)} \tilde{D}_v(\theta_{v[i]}, \theta_{v[j]}). \label{eq:sco-loss-batch}
\end{equation}
This allows an online form of the algorithm, similar to Elo, where ratings for players can be updated in a decentralized fashion after receiving the outcome of a single game ($|B| = 1$).

\subsubsection{Sigmoidal Programming}
\label{sec:sco-via-sp}

A different way to compute SCO ratings is via sigmoidal programming~\cite{Udell2014MaximizingAS}: solve for the (soft) optimum directly, \ie find $\btheta^*$ that minimizes
$\tilde{L}$ defined in equation~\ref{eq:sco-loss}.
Note the $\tilde{L}$ can be rewritten in terms of the number of pairwise interactions between agents, quantified in the $\bN$ matrix:
\begin{eqnarray}
\tilde{L}([\succ], A, V, \btheta) & = & \sum_{a, b \in A \times A} N(a,b) \sigma(\theta_b - \theta_a)
\label{eq:sco-loss-sum-with-n}
\end{eqnarray}
which is a sum of {\it sigmoidal functions} $\sigma$ defined as functions which are strictly convex on domain $\theta_b - \theta_a \le z$ and then strictly concave on $\theta_b - \theta_a \ge z$ (\ie $z = 0$).
With a variable per difference in pair of ratings and appropriate constraints on variables and their feasible regions, the resulting optimization problem is known as a {\it sigmoidal program} which can be solved using a branch-and-bound algorithm~\cite{Udell2014MaximizingAS}.
A detailed construction is presented in Appendix~\ref{app:sco-as-sp}\footnote{
All appendices are available in the technical report version of the paper~\cite{lanctot2024softcondorcetoptimizationranking}.}

\begin{algorithm}[t]
\DontPrintSemicolon 
\KwIn{A preference profile $[\succ]$}
\KwIn{An initial parameter vector $\btheta^0 = \left( \frac{\theta_{max} - \theta_{min}}{2} \right) \mathbf{1}$, where $\mathbf{1}$ is a vector of ones of length $|A|$} 
\KwIn{Learning rates for each step $\alpha^t$}
\KwIn{Batch size $K$}
\For{$t \in \{1, 2, \cdots, T\}$}{
  $B \gets$ $K$ votes sampled uniformly from $[\succ]$ \;
  Define $\nabla_{\btheta} \tilde{L}(B, \btheta)$ based on \eqref{eq:sco-loss-batch}. \label{alg:define-grad} \;
  $\btheta^t = \textsc{Proj}(\btheta^{t-1} - \alpha^t \nabla_{\btheta} \tilde{L}(B))$ \label{alg:update-line}
}
\Return{$\btheta^T$} \;
\caption{Learning SCO ratings by gradient descent}
\label{alg:sco}
\end{algorithm}

\subsection{Fenchel-Young Loss Optimization}
\label{sec:sco-fenchel-young}

Here, we give an overview of the implementation of Fenchel-Young loss optimization; for more detail on the precise formulation and relationship perturbed optimizers, please see Appendix~\ref{app:fy-loss-min}.
 
In practice, Fenchel-Young loss optimization follows Algorithm~\ref{alg:sco}, with a different definition of the loss and hence gradient on line~\ref{alg:define-grad}. For a single vote $v$, a stochastic version of the gradient $\hat g_v$ can be computed: let $\btheta_v$ be the ratings for agents compared in $v$.
Let $\bX$ be a vector of Gumbel-distributed random variables of size $|v|$. Then $\tilde{\btheta}_v = \btheta_v + \sigma X$ is the {\it perturbed ratings} and let $\textsc{ArgSort}(-\tilde{\btheta}_v)$ be the the indices of the elements that would sort the values.
The Fenchel-Young gradient is then:
\begin{equation}
\hat g_v = \textsc{ArgSort}(\textsc{ArgSort}(-\tilde{\btheta}_v)) - (0, 1, \cdots, |v|-1), 
\label{eq:fy-grad}
\end{equation}
for all agents in $v$, and 0 otherwise.
Inuitively, if the perturbed ranking obtained by sorting $\tilde{\btheta}_v$ would yield the same order of agents as in $v$, then the gradient for this vote would be zero. Otherwise, it is nonzero and each element of the gradient corresponds to the difference in rank position between the vote and perturbed ranking.
Similarly to standard gradient descent, these gradients can be accumulated over batches $|B| \geq 1$ as $\nabla_{\btheta} \tilde{L}^{FY}(B, \btheta) = \sum_{v \in B} \hat g_v$.

\subsection{Theoretical Properties}
\label{sec:sco-theory}


Given the SCO framework, defined through the loss function introduced in equation~\ref{eq:sco-loss}, the first question to ask is whether its solution does lead to rankings with desired properties. We answer this question in the affirmative.

\begin{theorem}
Given the sum of soft Kendall-tau distances:
\label{th:sigmoid-condorcet}
\begin{eqnarray}
\tilde{L}([\succ], A, V, \btheta) & = & \sum_{a, b \in A \times A} N(a,b) \sigma(\theta_b - \theta_a)\, ,
\end{eqnarray}
if for preference profile $[\succ]$, voters $V$, there exists a candidate $c \in A$ that is a Condorcet winner, the loss is monotonically decreasing with $\theta_c$. As a consequence, if $\btheta^*$ is a global minimum of $\tilde{L}$ on the $\ell_\infty$ ball of radius $\theta_{\max}$, then $\theta^*_c = \theta_{\max}$.
\end{theorem}

\begin{proof}[Proof (sketch -- for a full proof, please see Appendix~\ref{app:proofs})]
Let $c \in A$ be the Condorcet winner. 
The loss $\tilde{L}$ as expressed in equation~\ref{eq:sco-loss-sum-with-n} is expressable in terms of a constant $K$ (that does not depend on $\btheta$) and a sum of sigmoids multiplied by coefficients from $\bM$, by symmetry of the logistic function: $\sigma(x) = 1 - \sigma(-x)$. 
The second term is decomposable into contributions from comparisons to agent $c$, which is monotonically decreasing in $\theta_c$, and a sum which does not depend on $\theta_c$. As a result, increasing $\theta_c$ always decreases the loss, hence the minimum must correspond to $\theta^*_c = \theta_{max}$.
\end{proof}

Theorem~\ref{th:sigmoid-condorcet}  assumes that it is possible to find $\btheta^*$. The challenge is $\tilde{L}$ is nonconvex in its parameters, $\btheta$, and thus standard gradient descent is not guaranteed to find a global minimum.  However, the sigmoidal programming approach described in Section~\ref{sec:sco-via-sp} is guaranteed to find a point that approximately minimizes $\tilde{L}$ within a specified tolerance region, though the problem may take exponential time in the number ($\Omega(m^2)$) of variables. Furthermore, as we will show in Section~\ref{sec:eval}, stochastic gradient descent, while without any guarantees, tends to perform very well in practice. 

The ranking loss used by Fenchel-Young optimization method described in Section~\ref{sec:sco-fenchel-young} is convex and Lipschitz with respect to its parameters. Since the parameters correspond to the ratings themselves, the loss is also convex with respect to the parameters. Hence, assuming we restrict the ratings to a compact, convex set, there exists a global minimum that gradient descent via Fenchel-Young gradients is guaranteed to approach assuming standard learning rate conditions (e.g., square-summable, not summable).

\section{Empirical Evaluation}
\label{sec:eval}

We run experiments to demonstrate a number of properties of interest. In particular, we are interested in understanding how closely the rankings obtained from SCO approximate those obtained via Kemeny-Young, compare to outcomes returned by perturbed optimizers and Elo using several sources of data:

\begin{description}
\item[Example data.] This is a preference profile used in Section~\ref{sec:eval-warmup}, example similar to the one from Table~\ref{tab:example}. 
\item[PrefLib data.] These are examples from the voting literature on Wikipedia and on real data from elections, sports analytics, and others from PrefLib~\cite{Preflib_MaWa13a}.
Note that we restrict ourselves to the strictly-ordered incomplete (SOI) and strictly-ordered complete (SOC) data types in PrefLib, yielding a total of 12,680 preference profiles.
\item[Synthetic evaluation data.] The synthetic data is generated to resemble those coming from online matching data, such as from a gaming site or tournament, based on TrueSkill~\cite{TrueSkill06}.
Agents' true skill values are normally distributed and contests (matches) between them are generated such that each agent's individual performance is stochastic with mean centered at their skill level. The outcome of each match-up is then a sorted list of each player's performance in the match-up, equal to their true skill plus normally-distributed noise.
Generated data allows us to mimic the structured sparsity present in real online game-play data but also to compare results to actual ground truth rankings. 
\item[Diplomacy game data.] This data set is our largest challenge problem: an anonymized agent-vs-agent data set of 7-player Diplomacy games played on the webDiplomacy web site (\url{webdiplomacy.net}) between 2008 and 2019~\cite{lanctot2023evaluating} with $m = 52,958$ agents (players) and $n = 31,049$ votes (games). 
Each game outcome is a strict order between seven players; the goal is to find a ranking over agents that minimizes the average Kendall-tau distance to all the votes.
With these values, only $0.0011 \%$ of the margin matrix entries are nonzero.
\end{description}
We chose these data sets to show specific properties of SCO ratings: top-ranking Condorcet winners, PrefLib ranking data naturally capturing real human preferences, the online game regime evaluation systems are commonly deployed but with ground truth ratings, and finally a very large challenge human evaluation problem.

For Elo ratings, the majorization-minorization algorithm of Hunter (used by bayeselo~\cite{Coulom05bayeselo}) is used to efficiently compute the best fit to the evaluation data~\cite{Hunter04btmodels}, and the SigmoidalProgramming package~\cite{Udell20SP} to solve sigmoidal programs.
By default we use gradient descent to minimize the sigmoid loss (Algorithm~\ref{alg:sco}) to compute the SCO ratings, but also compare Fenchel-Young gradients and sigmoidal programming. 
Unless otherwise noted, $\theta_{min} = 0$ and $\theta_{max} = 100$.
Full details such as specific hyperparameter values, please see Appendix~\ref{app:results}.
We also show two additional experiments in Appendix~\ref{app:results}: one which
shows SCO used to approximate a Bayesian posterior over rating and one that shows SCO's online performance.

\subsection{Warmup: Top-Ranking Condorcet Winners}
\label{sec:eval-warmup}


Lanctot et al. showed that Elo assigns the same rating to agent $A$ and $C$ in the example in Table~\ref{tab:example}, despite $A$ not being a Condorcet winner~\cite{lanctot2023evaluating}.
In contrast, SCO is designed to find the optimal ranking according to Definition~\ref{def:optimal-ranking}, which top-ranks Condorcet winners when they exist. 
Consider the following 5-vote preference profile:
\begin{equation}
\label{eq:example}
2: A \succ B \succ C, \hspace{1cm} 3: C \succ A \succ B.
\end{equation}
Note that $\delta(C, A) = \delta(C, B) = 3 - 2 > 0$, hence agent $C$ is a strong Condorcet winner.
However, the win rate of $A$ ($\frac{7}{15}$) is higher than $C$ ($\frac{6}{15}$), hence Elo assigns strictly higher rating to agent A.
Since there are only six possible rankings, it is easy to verify that $\tilde{L}$ is minimized for the optimal ranking $C \succ A \succ B$. 
Since $n = 5$, we use full gradient descent (no batching) and compare to stochastic gradient descent with a batch size of 2.
In all cases, Algorithm~\ref{alg:sco} using the sigmoid loss converges to the optimal ranking, and so does sigmoidal programming. 

We also find that Fenchel-Young gradient descent top-ranks agent $A$. This is because the gradient of a rating using the Fenchel-Young loss is weighted by the difference in ranks rather than just order misclassifications like the soft Kendall-tau distance. We elaborate on this in Section~\ref{sec:discussion}.
Full results are shown in Appendix~\ref{app:full-warmup-results}. 

\subsection{Kemeny-Young Approximation Quality}
\label{sec:eval-kemeny-approx-qual}

\begin{table}[t]
    \centering
    \begin{tabular}{rr|rrr|rr|rr}
    $m_{\bot}$  & $m_{\top}$ & size & $\bar{m}$ & $\bar{n}$  & $C^{gd}$ & $C^{sp}$ & $\bar{K}_n^{gd}$ & $\bar{K}_n^{sp}$   \\
    \hline
         2     &    2      &   11     &    2.00   &  29    & 1.00 & 1.00 & 0 & 0  \\
         3     &    3      &   115    &    3.00   &  1878  & 1.00 & 1.00 & 0 & 0  \\
         4     &    4      &   162    &    4.00   &  7189  & 1.00 & 0.99 & 0.005 & 0.009 \\
         5     &    5      &   135    &    5.00   &  34666 & 1.00 & 0.66 & 0.024 & 0.039 \\
         6     &    6      &   109    &    6.00   &  33266 & 0.99 &      & 0.043 &      \\
         7     &    7      &   92     &    7.00   &  28755 & 0.97 &      & 0.029 &      \\
         8     &    8      &   73     &    8.00   &  18336 & 0.96 &      & 0.032 &    \\
         9     &    9      &   88     &    9.00   &   4190 & 0.94 &      & 0.027 &   \\
         10    &    10     &   80     &   10.00   &   3289 & 0.97 &      & 0.023 &   \\
         11    &    20     &   1721   &   16.48   &   127  & 0.99 & & &   \\
         21    &    50     &   1465   &   30.75   &    34  & 0.98 & & &   \\
         51    &    100    &   567    &   71.50   &    40  & 0.92 & & &   \\
         101   &    200    &   2989   &  124.00   &    25  & 0.99 & & &   \\
         201   &    500    &   4540   &  302.81   &    65  & 0.98 & & &   \\
         501   &    --     &   533    & 2190.04   &    50  & 0.56 & & &   \\
    \hline
    \end{tabular}
    \caption{Kemeny-Young approximation quality on 12,680 PrefLib instances, grouped by number of alternatives where $m_{\bot} \le |A| \le m_{\top}$. Each row corresponds to a group, size to the number of instances in each group, $\bar{m}$ and $\bar{n}$ are the average number of alternatives and votes in each group. $C^{gd}$ and $C^{sp}$ refer to the Condorcet match proportions for gradient descent and sigmoidal programming, respectively.
    Similarly, $\bar{K}_n$ refers to normalized Kendall-tau distance to the Kemeny ranking, averaged over all instances in the group.}
    \label{tab:preflib}
\end{table}

We evaluate approximation quality (compared to the Kemeny-Young ranking) of Algorithm~\ref{alg:sco} on PrefLib instances~\cite{Preflib_MaWa13a}.
We run Algorithm~\ref{alg:sco} with a batch size $|B| = 32$, learning rates $\alpha \in \{ 0.01, 0.1 \}$, iterations $T \in \{ 10^4, 10^5 \}$, and temperature $\tau \in \{1, \frac{1}{2}\}$ averaged across 3 seeds per instance on all 12,680 PrefLib data instances. On the instances where $|A| \le 10$ we also run the Kemeny-Young method. Denote an instance by $i \in \{1, 2, \cdots, 12680 \}$. 
Each produces a ranking which we denote $R_{i,\text{SCO}}$ and $R_{i,\text{Kem}}$. 

We compute two metrics:
(i) {\it Condorcet Match Proportion}: this is the proportion of instances that top-ranks a Condorcet winner when one exists. (ii) {\it Normalized Kendall-tau Distance}: For all instances $i$ with $|A| \le 10$, the average value of $K_n(R_{i,\text{SCO}}, R_{i,\text{Kem}})$.

We show the aggregated metrics for 15 groupings of alternatives ($m = |A|$) that partition the 12,680 preference profiles in Table~\ref{tab:preflib}.
Generally, when using gradient descent (Algorithm~\ref{alg:sco}, Section~\ref{sec:sco-via-sgd}) the Condorcet winner is top-rated when it exists 92\% - 100\% of the time when $|A| \le 500$ and the average normalized Kendall-tau distance to the Kemeny solution is low ($\bar{K}_n(R_{i,\text{SCO}}, R_{i,\text{Kem}}) \le 0.043)$.
We found that sigmoidal programming worked as well as Algorithm~\ref{alg:sco} for instances where $m \le 4$. On instances with five or more alternatives, there were numerical instabilities with the sigmoidal programming approach leading to a high number of failures. Hence, we run only gradient descent when $m \ge 6$.

\subsection{Sparse Data Regime}
\label{sec:eval-sparse-data}

In this subsection, we generate synthetic data by simulating evaluations from match-ups played in an online game setting. 
We do this in two ways: one that is uniform (reflecting a round-robin style tournament), and one that leads to a structured form of sparsity often encountered in competitive gaming (\ie skill-matching platforms).
We use $|A| = 20$ agents where for each agent $i$: $\theta_i \sim N(100, 30)$, and contests between 4 agents (\ie four-player games).
To generate contests, two separate distributions are used: the \defword{uniform distribution} samples agents uniformly at random, 
and the \defword{skill-matched distribution} which incrementally 
builds each contest, drawing 3 new candidates at random and choosing the one whose true rating is closest to the average of the set of agents so far.
Then, for each contest $c$, we simulate the performance of agent $i$ in that contest, $P_i(c) = \theta_i + \epsilon_{c,i}$, where each $\epsilon_{c,i} \sim N(0, 5.0)$. The outcome of the contest (vote among contestants) is then obtained by sorting the performances of all  the agents in that contest.

We run Algorithm~\ref{alg:sco}, Elo, and several VasE methods across many simulated $n$-contest tournaments, where each value of $n$ corresponds to a proportion of missing data $p_\emptyset$ (alternative pairs that have not been evaluated in a contest together), with $p^u_{\emptyset}$ and $p^s_{\emptyset}$ referring the the uniform and skill-match distributions respectively:

\vspace{0.2cm}

\begin{tabular}{|l|cccccccc|}
\hline
$n$                 & 5    & 10   & 20   & 30   & 50   & 75   & 100  & 200   \\
\hline
$p^u_{\emptyset}$   & 0.85 & 0.72 & 0.52 & 0.38 & 0.20 & 0.09 & 0.04 & 0.001 \\
\hline
$p^s_{\emptyset}$   & 0.88 & 0.75 & 0.59 & 0.49 & 0.36 & 0.28 & 0.23 & 0.15 \\ 
\hline
\end{tabular}

\vspace{0.2cm}

\noindent For each value of $n$ we run 200 instances using different seeds and report average values.
For each run, we used 10000 iterations with batch size 16.
We show two different metrics. First, the Kendall-tau distance between the final ranking found by Algorithm~\ref{alg:sco} and the true ranking given the true ratings (maximum value of $\frac{20 \cdot 19}{2} = 190$).
This first metric identifies the pairs of agents whose relative order disagree between SCO and the true ranking; we denote these discordant pairs $D_{SCO,true}$.
The second metric, which we call ``mean true ratings distance'' (MTRD), is then defined to be the average absolute difference in true ratings between all pairs of agents in these disagreements $\sum_{(i,j) \in D_{\text{SCO},\text{true}}} | \theta_i - \theta_j | / |D_{\text{SCO},\text{true}}|$, which allows us to take a nuanced look at the optimized parameters in addition to the associated ranking. The results are shown in Fig.~\ref{fig:noisy_sparse_tournament}.

\begin{figure}[t]
    \centering
    \includegraphics[width=\columnwidth]{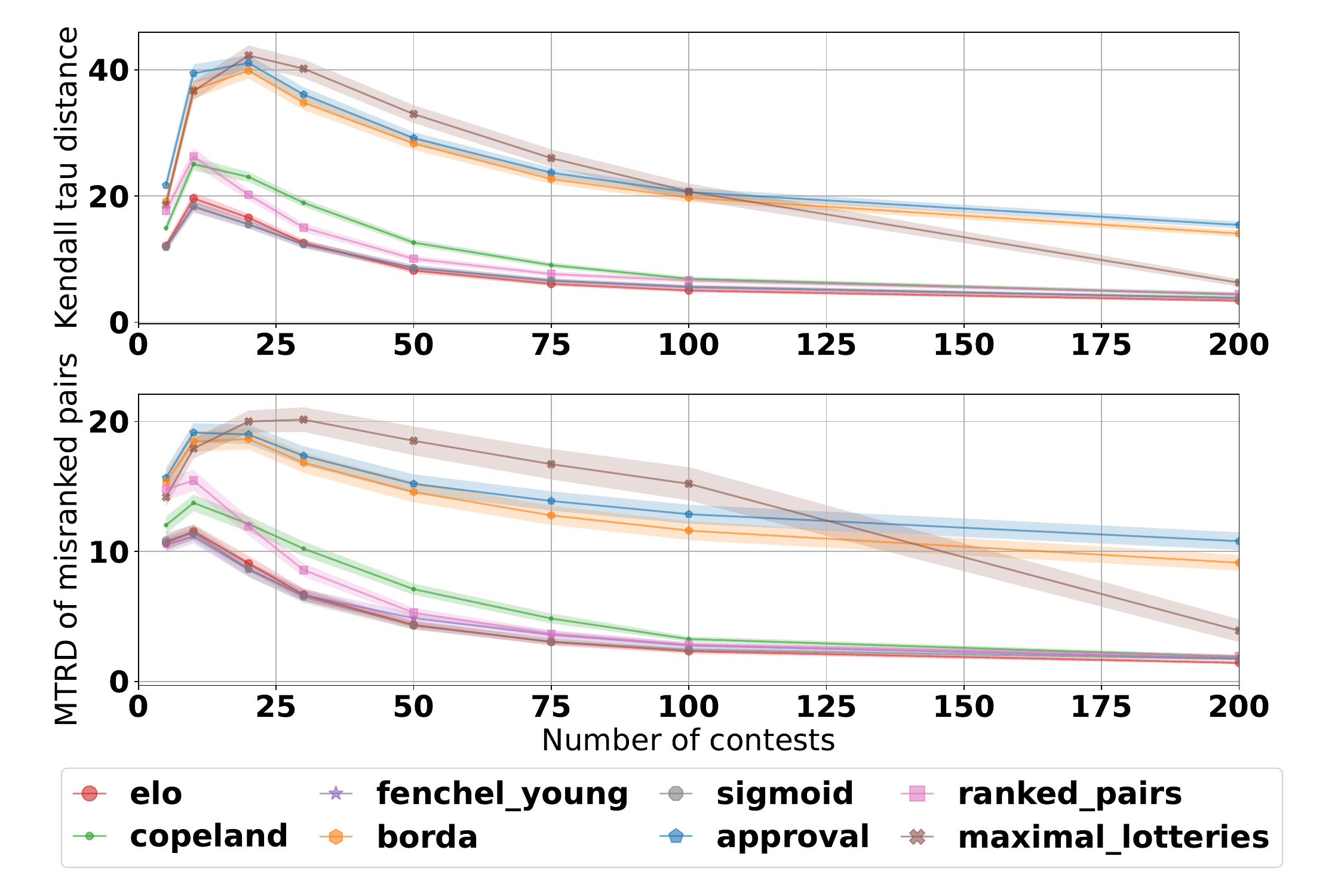}
    \includegraphics[width=\columnwidth]{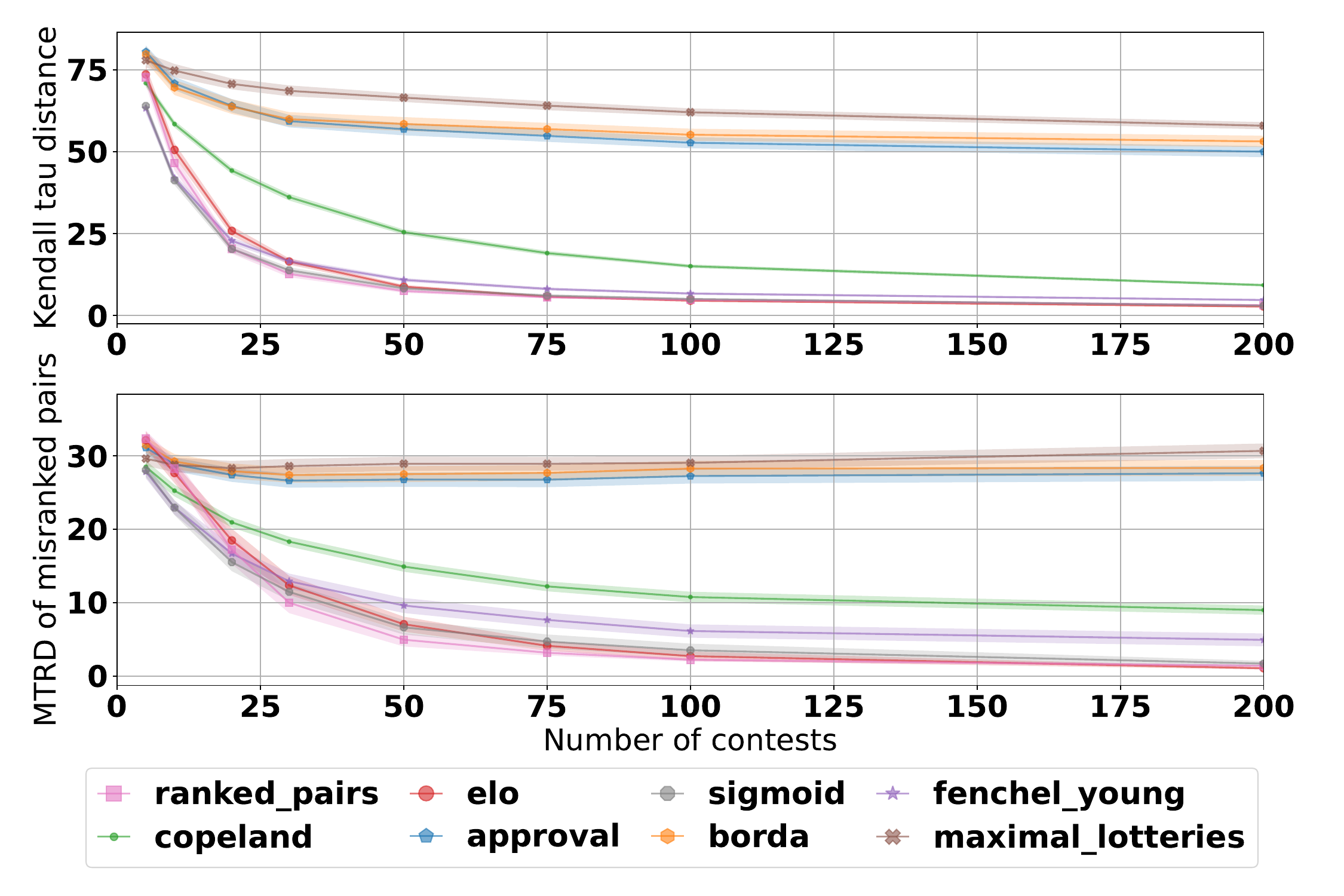}
    \caption{Kendall-tau distance (KTD) of ranking to true ranking, and Mean True Rating Distance (MTRD) of misranked pairs for ranking methods in tournament settings. 
    Top: uniform distribution, Bottom: skill-matched distribution.
    Error bars represent 95\% confidence intervals.}
    \label{fig:noisy_sparse_tournament}
\end{figure}

Of the VasE voting methods: approval, Borda, and maximal lotteries have the highest KTD and MRTD, and struggle especially in the skill-matched distribution. This is unsurprising; for example, approval and Borda will naturally weight alternatives according to their representation in the data.
Under both distributions, when 59\% or more of the match-ups are missing, SCO ratings (computed using both sigmoid and Fenchel-Young losses) achieve the lowest KTD and MRTD.
Under the uniform distribution, SCO achieves the lowest when 38\% or more of the data is missing.
Under the skill-matched distribution, ranked pairs achieves comparable KTD to SCO ratings, and lower MRTD when the amount of missing data is less than 50\%.
Elo values are comparable to SCO in the uniform distribution and  higher in the skill-matched distribution.

\subsection{Diplomacy Game Outcome Prediction}
\label{sec:eval-diplomacy}


In this subsection we investigate the question of how well SCO rankings predict human game outcomes. Recall that this data set consists of all human-played seven-player games taken from the {\tt webDiplomacy.net} spanning 11 years, resulting in a data set of size $m = 52,958$ and $n = 31,049$. As a result, less than $0.01\%$ of the unique player combinations are observed in the data.
Finding a single ranking that sufficiently predicts the unseen test data is a difficult due to the extreme sparsity of the preference data.

We reflect the evaluation used in~\cite{lanctot2023evaluating} with a small modification to adopt common practice in the supervised learning setting .
First, we create 50 random splits of the data into 
training sets $\cD_R$ and testing sets $\cD_T$, with $|\cD_R| = 28049$ game outcomes (votes) and $|\cD_T| = 3000$ game outcomes.
The random splits are such that each alternative in the test set is seen at least once on the training set, but no game outcomes (data points) are shared across train / test split.
At each iteration $t$ the method has a ranking denoted $R_t$ learned from data in $\cD_R$; we then compute and report the average Kendall-tau distance over the test set: $\textsc{KTD}_{test}(t) = \frac{1}{3000} \sum_{g \in \cD_T} K_d(g, R_t)$. 

The resulting $\textsc{KTD}_{test}(t)$ is shown in Figure~\ref{fig:sco-diplomacy}. We found that batch size had a minor effect on the results, so we show results for batch size 32 only.
SCO with sigmoid loss reaches an average Kendall-tau distance of $8.10$ after roughly 190000 iterations, and Fenchel-Young loss reaches $8.05$ at 600000 iterations.
There is an effect of increasing error after reaching this low point, likely due to overfitting; this could be reduced with early stopping, annealing learning rates, or other forms of regularization.
As a point of comparison, Elo and the best VasE method on this data set, VasE(Copeland), achieve a value of $8.34$. The next best VasE method was plurality, achieving a value of 8.57.
The more complex Condorcet methods, such as ranked pairs and maximal lotteries, cannot be run on this dataset due to their complexity, since $m = 52,958$.

\begin{figure}[t]
    \includegraphics[width=\columnwidth]{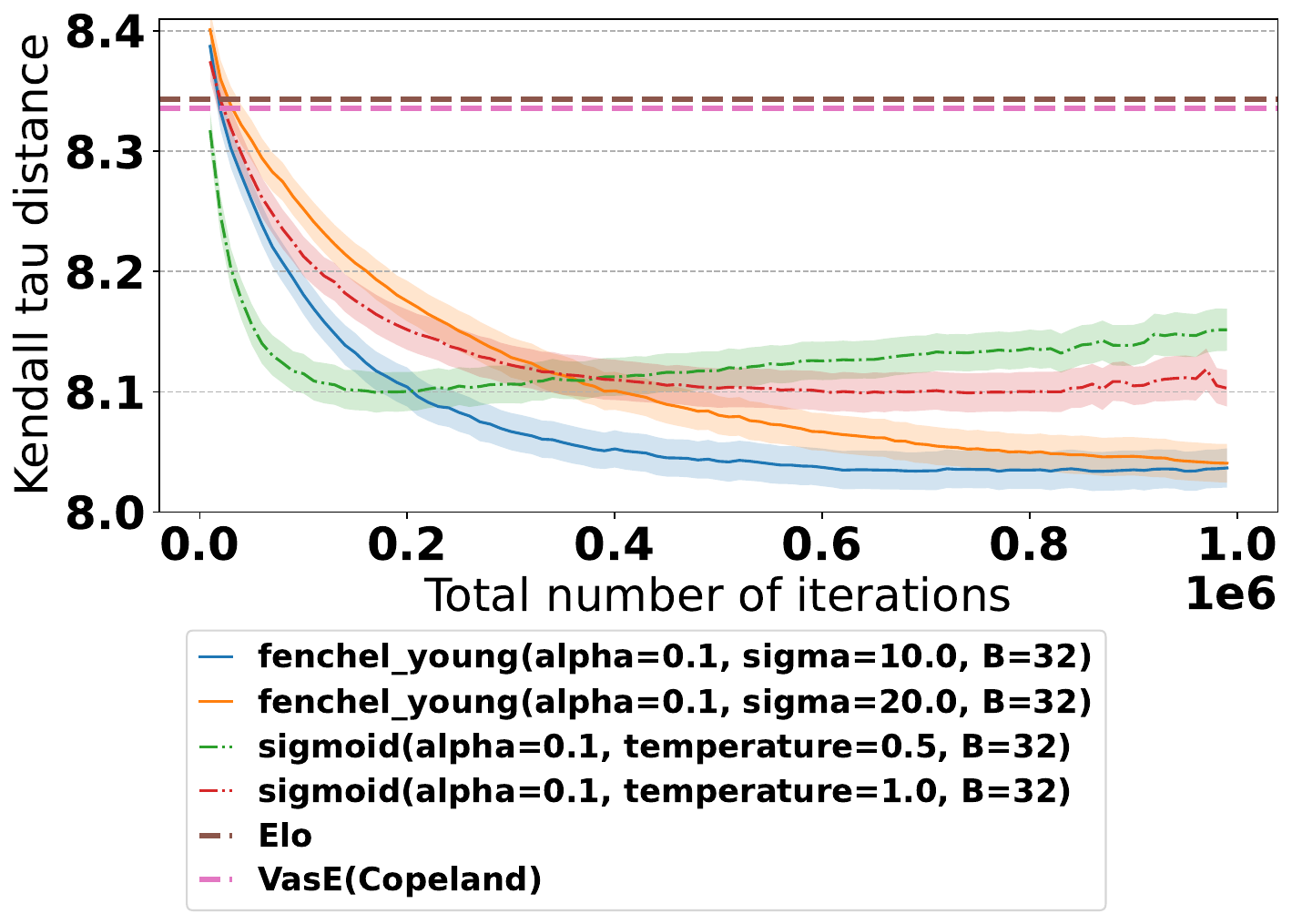}
    \caption{Average $\textsc{KTD}_{test}(t)$ between rankings and Diplomacy game outcomes across 50 seeds and train/test splits. All runs use batch size 32. Error bars depict 95\% conf. intervals.}
    \label{fig:sco-diplomacy}
\end{figure}

\section{Discussion}
\label{sec:discussion}


Which algorithm should be used to compute SCO ratings?
In our experience, sigmoidal programming produced similar results to standard gradient descent using the sigmoid loss; however, it sometimes suffered from numerical instability and was hard to scale to large number of agents due to programs requiring $m^2$ variables. Hence, we recommend using the sigmoidal programming approach only when there are a relatively low number of alternatives.
The Fenchel-Young loss is convex and hence gradient descent is guaranteed to converge to a global minimum; 
however, Condorcet winners (when they exist) are not necessarily top-ranked at that global minimum.
The practical performance of Fenchel-Young loss minimization is comparable to sigmoid loss minimization and slightly better in the large Diplomacy problem. 
Both the sigmoid and Fenchel-Young losses are minimized by gradient descent so can be optimized online (batch size $|B| = 1$) and work particularly well in approximating the optimal rankings when a large portion of the evaluation data is missing.

The differences between sigmoid and Fenchel-Young loss minimization is further illustrated in Figure~\ref{fig:loss-landscapes}, where we plot two landscapes of the Fenchel-Young and sigmoid losses for the same data over three agents, using the example of vote profiles in Equation~\ref{eq:example}. Since both losses are invariant by adding a constant, we plot over a 2D slice, for all values of $\theta$ with the same sum. As discussed above, the Fenchel-Young loss is strictly convex, its global minimizer exists is found by gradient descent. Like Elo, it favours the win rate and assigns the highest rating not to the Condorcet winner $C$, but to $A$. In contrast, the sigmoid loss is nonconvex, and has no global minimizer (it keeps decreasing at infinity). Optimizing over constrained rating, it assigns the highest rating to the Condorcet winner $C$.

If Condorcet-consistency or distance to the optimal ranking is important, we recommend using the sigmoid loss as it optimizes to minimize the distance to it directly; despite being non-convex, in practice it finds the Condorcet winner when it exists $\ge 96\%$ of the time when the number of alternatives $m \le 500$ and returns almost-optimal rankings on instances where it can be compared (Section~\ref{sec:eval-kemeny-approx-qual}).
Whereas, if assured convergence or weighting the loss functions by win rates (like Elo) is more important, we recommend using the Fenchel-Young loss optimization.

\begin{figure}[t]
    \includegraphics[width=\columnwidth]{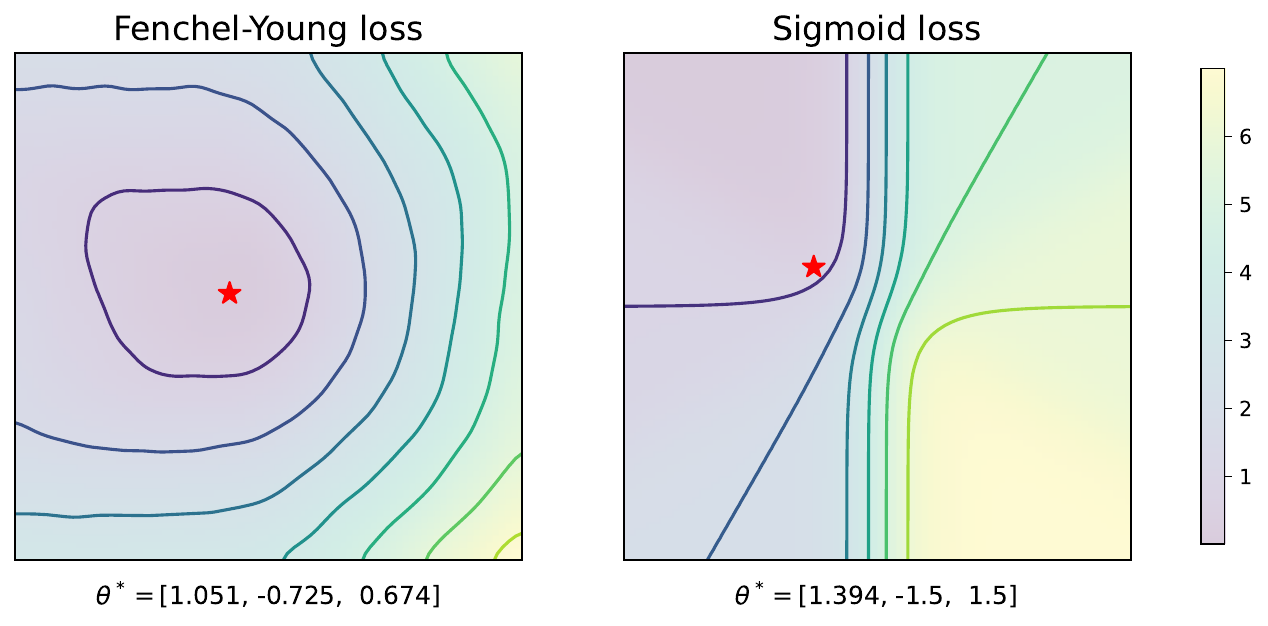}
    \caption{Loss landscapes for the Fenchel-Young (left) and sigmoid (right), with minimizer (red star) over $[-1.5, 1.5]^3$.}
    \label{fig:loss-landscapes}
\end{figure}

For future work, we would like to compare the performance of SCO to faster methods for finding Kemeny rankings (or approximations thereof)~\cite{KenyonMathieu07,Karpinski10,Ali12,Rico23Reducing,Rico23Kemeny} or other ranking methods inspired by tournament solutions~\cite{rajkumar15ranking}, or alternatives~\cite{Conati23}.
Considering similar differential approximations to other ranking distance functions such as Spearman's footrule distance~\cite{Diaconis77} as well as other ways to aggregate rankings~\cite{Dwork01} could be worthwhile.
Finally, we would like to investigate using SCO and social choice theory for driving post-training for alignment of language models~\cite{Conitzer24,swamy2024minimaximalist,maurarivero2025jackpotalignmentmaximallottery}.

\begin{acks}
We would like to that thank Luke Marris and anonymous reviewers of the Social Choice and Learning Algorithms Workshop at AAMAS 2024 for their valuable feedback on earlier versions of this paper.
\end{acks}


\clearpage
\bibliographystyle{ACM-Reference-Format} 
\bibliography{paper}


\onecolumn
\clearpage
\appendix

\section{Proofs of Theorems}
\label{app:proofs}

Recall that for a sigmoid $\sigma(x) = 1 - \sigma(-x)$ and the definitions of the $\bN$ and $\bM$ matrices from Section~\ref{sec:background-sct}.

First, we will introduce some helpful notation and state a few facts. For a set of alternatives $A = \{ a, b, \cdots \}$ and any ranking (total order) $R$ let $a \succ b \in A$ refer to the set of combinations where the left alternative is ranked higher the right alternative:
\[
\{ (a, b) \in A \times A ~|~ a \succ b \}.
\]
This ranking can be arbitrary and need not refer to one returned by a social welfare function; \eg for $A = \{ a_1, a_2, \cdots, a_m \}$, then one natural ranking is by numerical indices so that
\[
a \succ b \in A = \{ (a_1, a_2), (a_1, a_3),~\cdots~, (a_1, a_m), (a_2, a_3), (a_2, a_4),~\cdots`, (a_2, a_m), (a_3, a_4),~\cdots~, (a_{m-1}, a_m) \}.
\]
When the ranking is not clear from the context, assume any default ranking is used to enumerate all combinations of size two.

Recall that the soft Kendall-tau loss (equation~\ref{eq:sco-loss}) can be expressed as a sum of sigmoids multiplied by coefficients from the $N$ matrix (equation~\ref{eq:sco-loss-sum-with-n}). $\tilde{L}$ can be rewritten instead in terms of $M$:
\begin{eqnarray*}
\tilde{L}([\succ], A, V, \btheta)& = & \sum_{a,b \in A \times A} N(a,b) \sigma(\theta_b - \theta_a) \\
& = & \sum_{a \succ b \in A} [ N(a,b) \sigma(\theta_b - \theta_a) + N(b,a)\sigma(\theta_a - \theta_b) ] \text{\hspace{2cm} since } \forall a \in A, N(a, a) = 0\\
& = & \sum_{a \succ b \in A} [ N(a,b) \sigma(\theta_b - \theta_a) + N(b,a) (1 - \sigma(\theta_b - \theta_a) ] \text{\hspace{2cm} since } \sigma(x) = 1 - \sigma(-x)\\
& = & \sum_{a \succ b \in A} [ N(b,a) + (N(a,b) - N(b,a)) \sigma(\theta_b - \theta_a) ]\\
& = & \sum_{a \succ b \in A} N(b,a) + \sum_{a \succ b \in A} M(a,b) \sigma(\theta_b - \theta_a).\\
\end{eqnarray*}
We now restate and prove the Theorem~\ref{th:sigmoid-condorcet}.

\begin{customthm}{1}
Given the sum of soft Kendall-tau distances:
\begin{eqnarray}
\tilde{L}([\succ], A, V, \btheta) & = & \sum_{a, b \in A \times A} N(a,b) \sigma(\theta_b - \theta_a)\, .
\end{eqnarray}
If for preference profile $[\succ]$, voters $V$, there exists a candidate $c \in A$ that is a Condorcet winner, the loss is monotonically decreasing with $\theta_c$. As a consequence, if $\btheta^*$ is a global minimum of $\tilde{L}$ on the $\ell_\infty$ ball of radius $\theta_{\max}$, then $\theta^*_c = \theta_{\max}$.
\end{customthm}

\begin{proof}[Proof of Theorem~\ref{th:sigmoid-condorcet}]
Without loss of generality, we assume that $c \in A$ is the first element in some arbitrary order, and we rewrite the loss as
\begin{eqnarray}
\tilde{L}([\succ], A, V, \btheta) & = & K + \sum_{a \succ b \in A} M(a,b) \sigma(\theta_b - \theta_a),
\end{eqnarray}
where $K$ is a constant that does not depend on $\btheta$, as seen above. 
Since $c \succ b$ for all $b \in A$, we rewrite this loss as
\begin{eqnarray}
\tilde{L}([\succ], A, V, \btheta) & = & K + \sum_{b \in A} M(c,b) \sigma(\theta_b - \theta_c) + \sum_{a \succ b \in A_{-c}} M(a,b) \sigma(\theta_b - \theta_a)\, ,
\end{eqnarray}
where $A_{-c}$ is the set of all candidates  except $c \in A$. 
Note that the first term is a constant that does not depend on $\btheta$. 
The second term is monotonically decreasing in $\theta_c$, since $M(c, b) > 0$ for all $b \in A_{-c}$ since $c$ is a Condorcet winner, and $\sigma$ is an increasing function. Further, the third term does not depend on $\theta_c$. As a consequence, for any vector $\theta$, increasing $\theta_c$ while leaving all other ratings equal decreases the loss.

As a consequence, if $\btheta^*$ is a global minimum on the $\ell_\infty$ ball of radius $\theta_{\max}$, $\theta^*_c = \theta_{\max}$. Indeed, if we had $\theta^*_c < \theta_{\max}$, increasing $\theta^*_c$ would decrease the loss, leading to a contradiction. 
\end{proof}

This result shows that if there exists a Condorcet winner, there is no global minimum of $\R^{|A|}$ of the loss, since it is always better to increase $\theta_c$. It also shows that for any vector of ratings $\theta_{-c}$, the loss is smaller for $\theta_c \ge \max_{a \in A_{-c}} \theta_{a}$ than if $c$ did not have one of the highest ratings.

When there are constraints on $\btheta$, depending on the specific constraints, this can lead to several global minimizers. There can be several candidates with a rating of $\theta_{\max}$, even if only one of them is a Condorcet winner.

\section{Fenchel-Young Loss Optimization}
\label{app:fy-loss-min}

One of the core concepts of this method is the use of stochastic smoothing in the ranking function. We consider the function $\by^*$ that maps a vector $\btheta \in [-\theta_{\min}, \theta_{\max}]^m$ (where $m=|A|$) to the vector of its ranks - i.e. the vector $(0, 1, \ldots, m-1)$ permuted in the same order as the coefficients of $\theta$. Its perturbed version defined by $\by_\varepsilon^*(\theta) = \E[\by^*(\theta + \varepsilon Z)]$, for a vector of Gumbel distributed noise $Z$ and $\varepsilon>0$ is a smooth differentiable function, akin to the softmax function for one-hot argmax. It can be estimated without bias by Monte Carlo methods, averaging $\by^*(\theta + \varepsilon Z^{(i)})$, for user-generated $Z^{(i)}s$.

The Fenchel-Young loss associated to this function for a single data point of observed ranks $y$ is given by $L_\varepsilon(\theta, \by) = F_\varepsilon(\theta) - \by^\top \theta$, where $F_\varepsilon(\theta) = \E[\by^*(\theta + \varepsilon Z)^\top (\theta + \varepsilon Z)]$. Its gradient is given by
\[
\nabla_\theta L_\varepsilon(\theta, \by) = \by^*_\varepsilon(\theta) - \by\, .
\]
For a batch $B$ of size $|B|$ of votes on subsets $S_1,~\ldots~, S_{|B|}$ of candidates $A$ and associated vectors of ranks $y_1, ~\ldots~, y_{|B|}$ (each $y_i$ corresponding to observed partial rankings of elements in $S_i$), the stochastic gradient update is given by
\[
g_{FY} = \frac{1}{|B|} \sum_{i=1}^{|B|} \pi_{S_i} (y^*_\varepsilon(\theta_{S_i}) - y_i)\, ,
\]
where $\pi_{S_i}$ linearly maps the $|S_i|$ coefficients in $(y^*_\varepsilon(\theta_{S_i}) - y_i)$ to the corresponding coefficients in the $|A|$ coefficients of $\theta$. 

In practice, it can be estimated without bias, by replacing $y^*_\varepsilon(\theta_{S_i})$ by a stochastic approximation
\[
\tilde g_{FY} = \frac{1}{|B|} \sum_{i=1}^{|B|} \pi_{S_i} (y^*(\theta_{S_i} + \varepsilon Z^{i}) - y_i)\, .
\]
This loss, beyond having gradients that can be efficiently computed (requiring only to perform sorting on perturbed coefficients of $\theta$) is differentiable and convex in $\theta$ and can thus be efficiently over a convex set. Further, if the observed ranks $y_i$ are indeed generated from noisy sorting of a vector of ranks $\theta_0$, then the loss is minimized at $\theta_0$ (see \cite{Berthet20} details).

\section{Solving a Sigmoidal Program to Compute SCO Ratings}
\label{app:sco-as-sp}

Recall, from Section~\ref{sec:sco-via-sp}:

\begin{eqnarray*}
\tilde{L}([\succ], A, V, \btheta) & = & \sum_{a, b \in A \times A} N(a,b) \sigma(\theta_b - \theta_a)
\end{eqnarray*}
A sigmoidal program~\cite{Udell2014MaximizingAS} has the form:
\begin{eqnarray*}
\text{minimize}    & \sum_i f_i(x_i) \\
\text{subject to}  & \bA \bx \le \bb \\
                   & \bC \bx = \bd \\
                   & l \le x_i \le u, \\
\end{eqnarray*}
whose objective is a sum of {\it sigmoidal functions} $f_i(x)$, defined as functions which are strictly convex on domain $x \le z$ and then strictly concave on $x \ge z$ (\ie $z = 0$ in our case). Note that $N(a,b) \ge 0$ for all $a, b \in A$.

Each variable $x_i$ corresponds to $\theta_a - \theta_b$, for $a, b \in A, a \not= b$, so there are $|A|(|A|-1) = \Omega(m^2)$ variables. The difference is lower-bounded by the parameter value range, so $u = \theta_{max} - \theta_{min}$ and $l = -u$.

There are no inequality constraints so $\bA = 0$ and $\bb = 0$. There are two types of equality constraints that are represented by $\bC$ and $\bd$:
\begin{enumerate}
    \item Constraints of the form $x_i + x_j = 0$ representing the fact that $(\theta_b - \theta_a) + (\theta_a - \theta_b) = 0$.
    \item Constraints to encode transitivity of the form $(\theta_b - \theta_a) + (\theta_a - \theta_c) = \theta_b - \theta_c.$
\end{enumerate}

\section{Relationships between SCO and Prior Work}
\label{sec:relationships-to-prior-work}

\subsection{Relationship between SCO and Probabilistic Ranking}
\label{sec:probranking}

The problem of selecting an ordering that satisfies voters' preferences has also been studied in statistics and psychology through probabilistic models defined over the space of permutations~\citep{diaconis1988rankprob, marden1995rankdata, alvo2014rankprob}. In particular, distance-based models~\citep{fligner1986distancebased} assume that voters have a true, collective preference over alternatives, $v^*$, and that the probability of observing any individual preference, $v \in V\subseteq [\succ]$, is proportional to its distance from $v^*$. Thus, the probability of any individual vote $v \in V$ is
\begin{equation}
    p(v) \propto \text{exp}(-K(v, v^*)) \label{eqn:prob_model}
\end{equation}
where $K$ is some metric between rankings. 

Abusing notation slightly, let $p^*(v)$ and $p_{\btheta}(v)$ be the distance-based 
probability distributions induced by \eqref{eqn:prob_model} for true 
collective ranking $v^*$ and some ranking $v_{\btheta}$ induced by parameters 
$\btheta$. 
To minimize the KL divergence between the true and parameterized distributions  
\begin{equation}
    \min_{\btheta} D_{KL}(p^*(v) || p_\theta(v)),
\end{equation}
we maximize the likelihood of the rankings under the observed rankings or votes $v\in V$ by optimizing the MLE objective:
\begin{align}
    \min_{\btheta} \mathbb{E}_{v \sim p^*(v)} \left[K(v, v_{\btheta})\right].
\end{align}
When $K = K_{d}$ is the Kendall-tau distance and we apply the smooth approximation in \eqref{eq:sigmoid-discrepancy-func}, we recover the SCO objective:  
\begin{equation}
    \min_{\btheta} \frac{1}{|V|} \sum_{v \in V} \left[\sum_{(a,b) \in \binom{m}{2}} \sigma(\btheta_{a} - \btheta_{b})\right].
\end{equation}

Since SCO falls within the class of distance-based models, it inherits a number of desirable properties, including  \emph{strong unimodality}~\citep{critchlow1991probabilitymodels,alvo2014rankprob}. This property states that if a true  collective preference $v^*$ exists,  then the probability of any preference $v$ is non-increasing as it moves away from $v^*$, with the ramification being that the MLE retrieves this unique (modal) ranking.  On the other hand, for other probabilistic models, such as the pairwise comparison models which include Bradley-Terry and Elo~\citep{BradleyTerry52, Elo78, Hunter04btmodels}, strong unimodality is not guaranteed in general, and instead requires (strong) stochastic transitivity among pairwise individual preferences~\citep{critchlow1991probabilitymodels, alvo2014rankprob, shahwainwright2018stochastictransitivity}, a property that it not always observed in agent evaluation data~\citep{bertrand2023elolim,Boubdir2023-nz}.


\subsubsection{Elo and Probablistic Rankings} As mentioned, Elo can be interpreted as an example of a pairwise-comparison based probabilistic ranking model. Since Elo is heavily used in the agent-evaluation context, we believe it is useful to fully illustrate this connection. This also allows us to further make a clear distinction between SCO and Elo. 

A complete preference $v \in [\succ]$ over the set of alternatives $A$ contains $\binom{|A|}{2}$ pairwise comparisons. Let the $\binom{|A|}{2}$ Bernoulli-distributed random variables  $\boldsymbol{X} = \{X_1, \ldots, X_{\binom{|A|}{2}}\}$ denote the outcome of observing every pairwise comparison on $v$. Each random variable $X_k$ represents the probability that alternative $a \in A$ is (strictly) preferred over alternative $b$. Thus, every Bernoulli parameter $p_i$ models the probability $p(a \succ b)$ as $p(x_k = 1) = p_k$ and the converse preference $p(b \succ a)$ as $p(x_k = 0) = 1 - p_k$. Pairwise comparison models~\citep{BradleyTerry52, critchlow1991probabilitymodels} assume that preferences between pairs of alternatives relate to their ratings $\btheta$ as:
\begin{align*}
    p(a \succ b) &= p(x_k = 1) = p_k = \sigma(\btheta_a - \btheta_b)\\
    p(b \succ a) &= p(x_k = 0) = 1-p_k = 1 - \sigma(\btheta_a - \btheta_b) = \sigma(\btheta_b - \btheta_a)
\end{align*}

Given this framework and letting $p^*(\boldsymbol{X})$ model the true joint probability distribution of observing the $\binom{|A|}{2}$ pairwise rankings, and $p_{\btheta}(\boldsymbol{X})$ be the  joint probability distribution induced by ratings $\btheta$, minimizing the KL divergence
\begin{align}
    & \min_\theta D_{KL}\left(p^*(\boldsymbol{X}) \big\| p_{\btheta}(\boldsymbol{X}) \right)
\end{align}
and assuming the independence of the observed pairwise comparisons, is equivalent to minimizing the negative log-likelihood:
\begin{align}
    \min_{\btheta} \sum_{k=1}^{\binom{|A|}{2}}\mathbb{E}_{x_k \sim p^*(X_k)}\left[-\log p_{\btheta}(x_k)\right]
\end{align}

When optimized over some set of votes  $v \in V\subseteq [\succ]$ this  leads to the objective:
\begin{equation}
    \min_{\btheta} \frac{1}{|V|} \sum_{v \in V} \sum_{i=1}^{\binom{|A|}{2}}\mathbb{E}_{x_k \sim p^*(X_k)}\left[-\log p_{\btheta}(x_k)\right].
\end{equation}




\subsection{Relationship between SCO and Elo}
\label{sec:relationship-to-elo}

Before we show the relation between SCO and Elo, we first summarize a well-known relationship between Elo and logistic regression~\cite{Morse19}. 

Suppose player $i$ plays player $j$ and the outcome is $y_{ij} = 1$ if $i$ won and $0$ otherwise. Recall from that Elo's online update rule is:
\begin{equation}
\theta_i \leftarrow \theta_i + K (y_{ij} - p_{ij}),
\label{eq:elo-online-update}
\end{equation}
where
\[p_{ij} = \frac{1}{1 + 10^{(\theta_j - \theta_i)/400}}\]
is Elo's predicted probability (of agent $i$ beating agent $j$), and $K = 32$~\cite{Elo78}. 
Logistic regression is a method for classification that models the prediction as a logistic function of a linear function between parameters and features~\cite{Cramer02,pml1Book}. In binary classification, data points $(\bx, y)$ comprise input vectors $\bx$ and classes $y \in \{ 0, 1 \}$. Define the log odds of the data point given a label of 1:
\[
\textsc{LogOdds}(\bx, y=1) = \log_{10} \left( \cdot \frac{\Pr(\bx~|~y=1)\Pr(y=1)}{\Pr(\bx~|~y = 0)\Pr(y=0)} \right).
\]

The probability of class 1 for a given $\bx$ is (via Bayes' rule) is
\begin{eqnarray*}
\Pr(y = 1~|~\bx) & = & \frac{\Pr(\bx | y = 1)\Pr(y=1)}{\Pr(\bx | y = 1)\Pr(y=1) + \Pr(\bx | y = 0)\Pr(y=0)}\\
& = & \frac{1}{1 + \frac{\Pr(\bx | y = 0)\Pr(y=0)}{\Pr(\bx | y = 1)\Pr(y=1)}}\\
& = & \frac{1}{1 + 10^{\textsc{LogOdds}(\bx, y=1) / 400}}.
\end{eqnarray*}
Let $p_k$ be the probability predicted by the model for data point $(\bx_k, y_k)$. The standard loss (error) used for a data point $(\bx_k, y_k)$ is a log loss $\ell_k = -\log p_k$ if $y_k = 1$ and $-\log(1 - p_k)$ if $y_k = 0$, which is combined into a single loss term:
\begin{equation}
\ell_k = -y_k \log p_k - (1 - y_k) \log (1 - p_k), 
\label{eq:log-loss}
\end{equation}
and the loss over a data set of $K$ examples is $\ell = \sum_{k = 1}^{K} \ell_k$.
The prediction is modeled using parameters $\btheta$ such that $p_k = \sigma(\btheta^T \bx_k)$. The derivative of the sigmoid is conveniently $\frac{d}{dx} \sigma(x) = \sigma(x) (1 - \sigma(x))$, so the gradient of the loss for point $k$ turns out to be:
\begin{equation}
\nabla_{\btheta} \ell_k = (\sigma(\btheta^T \bx_k) - y_k) \bx_k.
\label{eq:logistic-regression-gradient}
\end{equation}
Here, $\btheta$ represents the Elo ratings and the data points (game outcomes), $\bx_k$, contain all zeros except 1 for the player who won and -1 for the player who lost.
Now, the gradient descent update rule is:
$\btheta \leftarrow \btheta - \alpha \nabla_{\btheta} \ell_k$; the update in  Equation~\ref{eq:logistic-regression-gradient} is proportional to Equation~\ref{eq:elo-online-update}. To see this, when $i$ beats $j$ and $y_k = 1$, then
\begin{eqnarray*}
-\alpha \nabla_{\btheta} \ell_k & = & (\sigma(\btheta^T \bx_k) - 1) \bx_k \\
& = & - \alpha(\sigma(\theta_i - \theta_j) - 1) \bx_k \\
& = & + \alpha (1 - \sigma(\theta_i - \theta_j)) \bx_k,
\end{eqnarray*}
and due to the definition of $x_k$ will result in the magnitude being added to $\theta_i$ and subtracted from $\theta_j$.

How, then, does SCO relate to Elo?
As in SCO, the Elo ratings are the parameters of the model and the loss is a function of the difference in ratings passed through a logistic function. 
The key difference is the loss function used.
In Elo, the gradient update is derived from fitting a classifier to predict the probability of a win given two agents' ratings. As such, the standard negative log-likelihood is used~\cite{pml1Book}.
In SCO, the loss is the sum of sigmoids represented as $D_v$ (\eqref{eq:sco-loss} and \eqref{eq:sigmoid-discrepancy-func}) itself, specifically because it is interpreted as a (smooth) Kendall-tau distance-- the optimal point being its minimum-- rather than as a prediction.

\subsection{Relationship to Perturbed Optimizers}
\label{sec:relationship-to-perturbed-optimizers}

%

A variety of methods have been considered to produce smooth proxys of ranking operators, i.e. functions that approximate a ranking or sorting function, that are differentiable (see, e.g. \cite{Blondel20} and references therein). Indeed ranking is just one type of discrete-output operation that can be efficiently made differentiable, and we describe here the approach of \cite{Berthet20} following their notation: consider a general discrete optimization problem over a set $\cY$ that can be parameterized by an input $\btheta \in \R^m$,
for a finite set of distinct points $\cY \subset \R^m$
with $\cC$ being its convex hull, of the form:
\begin{equation}
\label{eq:max-argmax-optimization}
F(\btheta) = \max_{\by \in \cC} \by^T \btheta, \mbox{~~and~~} \by^*(\btheta) = \argmax_{\by \in \cC} \by^T \btheta.
\end{equation}
In particular, consider the case where $\cC$ is a convex polytope and the problems are linear programs. As an example, the problem of ranking a vector of $d$ entries can be cast in this fashion, by taking $C$ to be the permutahedron $\cP(\rho)$ (see below).  
The function $\by^*: \R^d \to \cY$ is piecewise constant, leading to zero or undefined gradients.
It can be smoothed by taking the mean under additive stochastic noise to the parameters: we define this proxy by the {\it perturbed maximizer}~\cite{Berthet20} for $\varepsilon > 0$ \begin{equation}
\label{eq:perturbed-argmax}
\by^*_{\varepsilon}(\btheta) = \bE_Z[y^*(\btheta + \varepsilon Z)] = \bE_Z[\argmax_{y \in \cC} \by^T (\btheta + \varepsilon Z)]\, ,
\end{equation}
where $Z$ is a random variable with a positive and differentiable density. Under these conditions, the argmax inside the expectation is almost surely unique, and $\by^*_\varepsilon$ is differentiable, with derivatives that are also expressed as a mean, and that can be estimated with Monte-Carlo methods (see \cite{Berthet20}).
Further, one can define the probability distribution $p_{\btheta}(y) = \Pr(\by^*(\btheta + \varepsilon Z) = y)$. It holds that $\by^*_{\varepsilon}(\btheta) = \bE_{p_{\btheta}(y)}[Y]$, and the {\it perturbed maximum} is
\begin{equation}
F_\varepsilon(\btheta) = \bE[F(\btheta + \varepsilon Z)] = \bE[\by^*(\btheta + \varepsilon Z)^\top (\btheta + \varepsilon Z)].
\end{equation}

The ranking operation for some vector of values $\bv$ can be seen as \textsc{Ranking}$(\bv) = \textsc{ArgSort}(\textsc{ArgSort}(\bv))$, where $\textsc{ArgSort}$ is a function that returns the indices that would sort the values. It is precisely this non-differentiable function that will be perturbed. For example, for 
$\bv = ( 0.18, -1.28, 0.65, 1.25, 0.25, -0.12 )$,
\[
\textsc{Ranking}(\bv) = (2, 0, 4, 5, 3, 1).
\]
Ranking can be cast as an LP in the form required by \eqref{eq:max-argmax-optimization}.
Let $\cP(\bw)$ be the {\it permutahedron}: the convex hull over all vectors that are permutations of $\bw \subset \R^m$. For $\rho = (m, \cdots, 1)$, a (linear) objective function is constructed that is the dot product of $-\btheta$ and $\by \in \cP(\rho)$.
The fundamental theorem of linear programming ensures that the solution of an LP is vertex of it polytope, unique for almost all inputs, corresponding to a permutation in the case of a permutahedron~\cite{Berthet20}.
As an example, suppose the ratings for agents $1, 2, 3, 4$ are $(10, 20, 30, 40)$. Then the maximum objective would be when $\by = \rho$ with a maximal value of $(4, 3, 2, 1) \cdot (-10, -20, -30, -40)$.
These methods learn such ranking functions from data via gradient descent. Indeed $\by_\varepsilon^*(\btheta)$ can be plugged in directly into any loss $\ell$ comparing it to a vector of ranks $\by_{\text{true}}$. Further, the formalism of {\it Fenchel-Young losses}~\cite{blondel2020learning} can be applied to this setting: they are particularly convenient to use with perturbed optimizer, as their gradient is given by
\[
\nabla_\theta \ell_{\text{FY}}(\theta; \by_{\text{true}}) = \by_\varepsilon^*(\theta) - \by_{\text{true}}\, .
\]

As detailed in our main text, this can be extended naturally to partial rankings, and constructing stochastic gradients based on batches of these rankings. The connection to SCO is clear. In perturbed optimizers, the user explicitly chooses the distribution of $Z$ (\eg Gumbel or Normal) and fits parameters $\btheta$ to the data under the model given by \eqref{eq:perturbed-argmax}.
In Young's formalization of Condorcet's model, the noise follows the Mallows model where the probability of a data point is proportional to its Kendall-tau distance to the ground truth rankings.

\section{Additional Results and Hyper-Parameter Values}
\label{app:results}

\subsection{Full SCO Results on the Failure Mode of Elo}
\label{app:full-warmup-results}

Table~\ref{tab:warmup-experiment} shows the full results for of all hyperparameters used for the warmup experiment described in Section~\ref{sec:eval-warmup}.

\begin{table}[t!]
\begin{tabular}{llll}
\hline
Full vs. Stochastic Gradient Descent & $\alpha$ & $\tau$ & $T_{conv}$ \\
\hline
GD & $0.01$ & $0.5$ & $289$ \\
GD & $0.01$  & $1$ & $1158$ \\
GD & $0.01$  & $2$ & $4661$ \\
GD & $0.1$ & $0.5$ & $28$ \\
GD & $0.1$  & $1$ & $115$ \\
GD & $0.1$  & $2$ & $463$ \\
\hline
SGD ($|B| = 2$) & $0.01$ & $0.5$ & $17.0$ \\
SGD ($|B| = 2$) & $0.01$  & $1$ & $1003.33$ \\
SGD ($|B| = 2$) & $0.01$  & $2$ & $4778.33$ \\
SGD ($|B| = 2$) & $0.1$ & $0.5$ & $25.00$ \\
SGD ($|B| = 2$) & $0.1$  & $1$ & $60.66$ \\
SGD ($|B| = 2$) & $0.1$  & $2$ & $401.66$ \\
\hline
\end{tabular}
\caption{Number of iterations to reach convergence to the optimal ranking on the example in Equation~\ref{eq:example}. $\alpha$ is the learning rate and $\tau$ is the temperature of the logistic function (Equation~(\ref{eq:sigmoid-discrepancy-func})).
$T_{conv}$ is the number of iterations required to converge to the optimal ranking; for SGD this is an average over 3 seeds.
\label{tab:warmup-experiment}}
\end{table}

\subsection{Kemeny-Young Approximation Quality}
\label{sec:kemeny-approx-qual}

For the results listed across PrefLib instances from Table~\ref{tab:preflib}: 
\begin{itemize}
    \item For sigmoidal programming, we used {\tt maxiters} of 10000. The tolerance parameter was first set to 0.0001. If a numerical instability occurred, then the problem was retried with a tolerance of 0.001 and then if that failed again, then 0.01 was used.
    \item For instances where $m \le 10$, SGD with the sigmoid loss used batch size of 32, 10000 iterations, $\alpha= 0.01$, and temperature $\tau = 1$.
    \item For instances where $11 \le m \le 500$, SGD with the sigmoid loss used batch size 32, 100000 iterations, $\alpha = 0.01$, and temperature $\tau = 1$.
    \item For instances where $m \ge 501$, SGD with the sigmoid loss used batch size 32, 100000 iterations, $\alpha = 0.1$, and temperature $\tau = \frac{1}{2}$.

\end{itemize}

\subsection{Approximation of Bayesian Posterior}
\label{app:eval-bayesian-posterior}

\begin{figure}[t]
    \centering
    \includegraphics[width=0.3\textwidth]{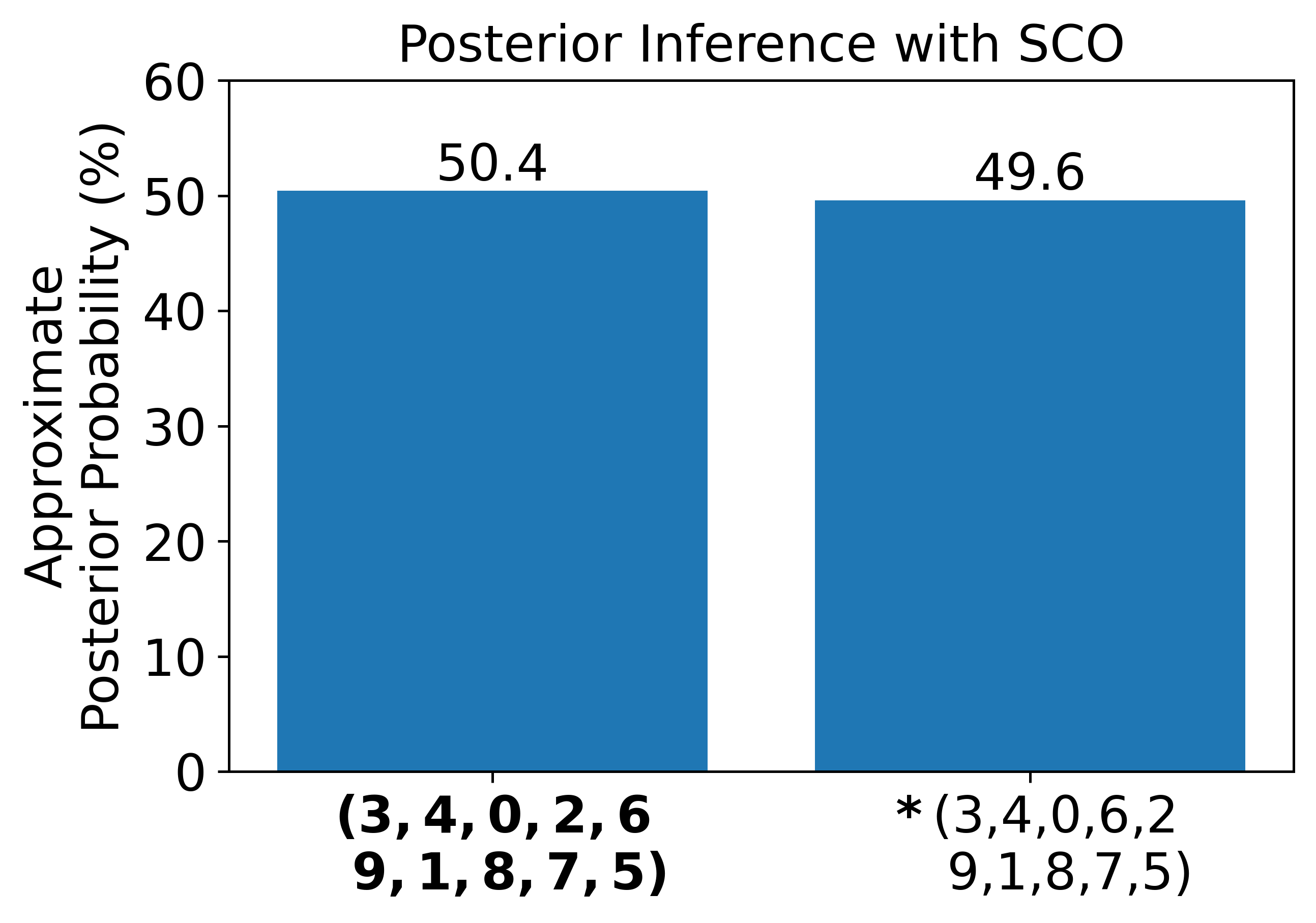}
    \caption{Posterior inference via Algorithm~\ref{alg:sco} on a synthetic evaluation problem. Contests between $4$ participants are randomly drawn from a pool of $10$. The ranking of the final rating $\theta_{T}$ is denoted with an asterisk (\textbf{*}). The true skill ranking is labeled in bold. True skills ratings, to two decimal places, are $[126.46, 106.00, 114.68, 133.61, 128.01, 85.34, 114.25, 97.73, 98.45,$ $106.16]$. All rankings sampled from the SGD-as-posterior-inference process are labeled in the bar plot with their corresponding observed frequencies as probabilities displayed above them. Notice that the mode of the posterior matches the true skill ranking and improves upon the final iterate.}
    \label{fig:sco_posterior}
\end{figure}


Prior work~\citep{mandt2016variational,stephan2017stochastic} has shown how stochastic gradient descent (SGD) can be modified to go beyond an MLE point estimate and approximately generate samples from the Bayes posterior of a probabilistic model. 
Under assumptions that SGD has converged to a region around a local minimum where the loss is quadratic and gradients are normally distributed, a new optimal step size can be computed from gradient samples. 

Here, we show this using Algorithm~\ref{alg:sco} with a constant learning rate, on simulated rank data generated similarly to the skill-matched distribution in Section~\ref{sec:eval-sparse-data}, but with $m = 10$. Once SCO has converged to a region, a new optimal step size is computed using gradient samples.
As SGD continues with this new step size, it traverses the stationary distribution of a Markov chain and produces iterates that well approximate samples from the true posterior distribution (see Figure~\ref{fig:sco_posterior}). While the posterior over ratings $\btheta$ is assumed Gaussian, the induced distribution over \emph{rankings} could be multimodal. Moreover, probabilities over rankings can provide uncertainty estimates to specific alternatives in the rankings.

Figure~\ref{fig:sco_posterior} displays the posterior over rankings. In this case, the true ratings for agents $2$ and $6$ are $114.68$ and $114.25$ respectively. These are the most closely rated agents in the game. The posterior distribution reflects the uncertainty in how to rank these two agents when observing noisy outcomes of their contests.

\subsection{Online Performance}
\label{app:online-eval}

We now evaluate the performance of SCO in the online setting. In the online setting, the data is presented to each method as a stream: one data point at a time. The ranking method must update its parameters incrementally using only the one data point.
The online form of SCO is obtained simply by running one stochastic gradient descent update on the data point (batch of size 1).
One advantage of Elo is that its online update involves updating only the ratings of players involved in the game that was played, so it can be employed in a decentralized way after each game.
SCO shares this advantage because the gradients of the batch $\nabla_{\btheta} L(B)$ involve only those players that were present in $B$.
For SCO, we use the same setup and evaluation metric as in Section~\ref{sec:eval-diplomacy}, except using batch size 1 and only one single pass over training data points, leading to a total number of iterations $T = |\cD_R| = 28049$. 
We run 50 such experiments with different seeds corresponding to 50 different orderings of the data.
The results are shown in Figure~\ref{fig:online_elo_vs_sco}.

\begin{figure}[t]
    \centering
    \includegraphics[scale=0.6]{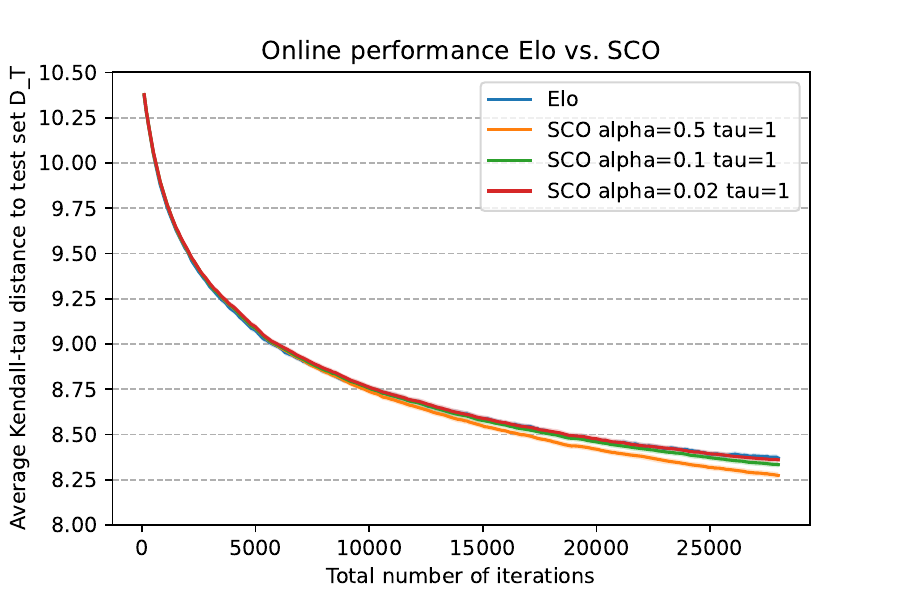}
    \caption{Online performance of SCO. Bands around each line represent 95\% confidence intervals over 50 runs.}
    \label{fig:online_elo_vs_sco}
\end{figure}

We tried $\alpha \in \{ 0.5, 0.2, 0.1, 0.02, 0.01 \}$ and temperatures $\tau \in \{ 0.5, 1, 2 \}$. The results were quite comparable, so we show only the extremes in variation for $\tau = 1$. In this case SCO with $\alpha = 0.02$ most closely resembles Elo while SCO with $\alpha = 0.1$ preforms slightly better than Elo and $\alpha = 0.5$ significantly better.

\end{document}